\documentclass[twocolumn,10pt]{IEEEtran}
\usepackage{algorithm,amsbsy,amsmath,amssymb,epsfig,bbm,mathrsfs,fancyhdr,multirow,stmaryrd,amsthm,nopageno,wrapfig}
\usepackage{caption}
\usepackage{color}
\usepackage{epstopdf}
\usepackage[noend]{algpseudocode}
\allowdisplaybreaks
\makeatletter
\def\BState{\State\hskip-\ALG@thistlm}
\makeatother

\theoremstyle{definition}

\newtheorem{myrem}{Remark}
\newtheorem{mypro}{Proposition}
\newtheorem{mylem}{Lemma}
\newtheorem{mythe}{Theorem}
\algrenewcommand\algorithmicwhile{\textbf{while}}
\algrenewcommand\algorithmicend{\textbf{end}}

\begin{document}
\title{Downlink Power Control in Two-Tier Cellular Networks with Energy-Harvesting Small Cells as Stochastic Games}
\author{Tran Kien Thuc, Ekram Hossain, and Hina Tabassum
\thanks{The authors are with the Department of Electrical and Computer Engineerng at the University of Manitoba, Canada (emails: trant347@myumanitoba.ca, \{Ekram.Hossain,Hina.Tabassum\}@umanitoba.ca). This work was supported by the Natural Sciences and Engineering Research Council of Canada (NSERC). A preliminary version of this work appeared in~\cite{tran}.}
}

\maketitle 

\begin{abstract}

Energy harvesting in cellular networks is an emerging technique to enhance the sustainability of power-constrained wireless devices. 
This paper considers the  co-channel deployment of a macrocell overlaid with small cells. The small cell base stations (SBSs)  harvest energy from environmental sources whereas the macrocell base station (MBS) uses conventional power supply. Given a stochastic energy arrival process for the SBSs, we derive a power control policy for the downlink transmission of both MBS and SBSs such that they can achieve their objectives  (e.g., maintain the signal-to-interference-plus-noise ratio (SINR) at an acceptable level) on a given transmission channel. We consider a centralized energy harvesting mechanism for SBSs, i.e.,  there is a central energy storage (CES) where energy is harvested and then distributed to the SBSs. When the number of SBSs is small, the game between the CES and the MBS is modeled as a single-controller stochastic game and the equilibrium policies are obtained as a solution of a quadratic programming problem. However, when the number of SBSs tends to infinity (i.e., a highly dense network), the centralized scheme becomes infeasible, and therefore, we use a mean field stochastic game to obtain a distributed power control policy for each SBS. By solving a system of partial differential equations, we derive the power control policy of SBSs given  the knowledge of mean field distribution and the available harvested energy levels in the batteries of the SBSs.   
\end{abstract}


\section{Introduction}
Energy harvesting  from environment resources (e.g., through solar panels, wind power, or geo-thermal power) is a potential technique to reduce the energy cost of operating the base stations (BSs) in emerging  multi-tier cellular networks. While this solution may not be practically feasible for macrocell base stations (MBSs) due to their high power consumption and stochastic nature of energy harvesting sources, it is appealing for  small cell BSs (SBSs)   that typically consume less power~\cite{deruyck}. Providing grid power to all SBSs may not always be feasible due to their
possible outdoor/remote/hard-to-reach locations. Wireless energy harvesting thus enables dense deployment of SBSs irrespective of the availability of grid power connections. In general, wireless energy harvesting can be classified into the following two categories: ambient energy harvesting and dedicated energy harvesting. In the former case, energy harvested from renewable energy sources (such as thermal, solar,  wind) as well as the energy harvested from the radio signals in the environment can be sensed by energy-harvesting receivers. In the latter case, energy from dedicated sources is transmitted to energy-harvesting devices to charge them wirelessly.

Designing efficient power control policies with different objectives (e.g., maximizing system throughput) is among one of the major challenges in energy-harvesting networks. 
In \cite{berk}, the authors proposed an offline power control policy for two-hop transmission systems assuming  energy arrival information at the nodes. The optimal transmission policy was given by the directional water filling method. In \cite{ding}, the authors generalized this idea to the case where many sources supply energy to the destinations using a single relay. A water filling algorithm was proposed to minimize the probability of outage. Although the offline power control policies  provide an upper bound and heuristic for online algorithms, the knowledge of energy/data arrivals is required which may not be feasible in practice.  In \cite{nicolo}, the authors proposed a two-state Markov Decision Process (MDP) model for a single energy-harvesting device considering random rate of energy arrival and different priority levels for the data packets. The authors proposed a low-cost balance policy to maximize the system throughput by adapting the energy harvesting state, such that, on average, the harvested and consumed energy  remain balanced. Recently, in \cite{harpreet}, the outage performance analysis was conducted for a multi-tier cellular network in which all BSs are powered by the harvested energy.  
A detailed survey on energy harvesting systems can be found in \cite{gunduz} where the authors summarized the current research trends and potential challenges.
 
Compared to the existing literature on energy-harvesting systems, this paper considers the power control problem for downlink transmission in two-tier macrocell-small cell networks considering stochastic nature of the energy arrival process at the SBSs. In particular, we assume that  ambient energy harvesting is exploited at a central energy storage (CES) from where energy can be transferred to the SBSs, for example, by using dedicated power beacons (PBs). PBs are low-cost devices that can potentially charge wireless terminals by transferring energy in a directional manner.  Note that the power control policies at the MBS and the SBSs and their resulting interference levels directly affect the overall system performance. The  design of efficient power control policies is thus of paramount  importance. 

In the above context, we formulate a discounted stochastic game model in which  all SBSs form a coalition to compete with the MBS in order to achieve the target signal-to-interference-plus-noise ratio (${\rm SINR}$) of their users through transmit power control.  That is, the MBS and the CES (which actually represents the set of SBSs in the game) compete to achieve the desired ${\rm SINR}$ targets of macrocell users and small cell users, respectively. Note that both the MBS and the SBSs transmit in the same channel (i.e., a co-channel deployment scenario is considered). Therefore, the competition (or conflict) arises due to the resulting cross-tier interference, i.e., as the MBS uses more power/energy to increase the utility of macrocell users, it results in higher cross-tier interference to small cell users. Similarly, the more energy/power the CES assigns to SBSs,  
the higher would be the cross-tier interference to macrocell  users. Note that the energy harvesting component is an important factor that indirectly contributes to the conflict. If the energy arrival rate is large, the SBSs will have a larger energy pool to spend and thus cause more interference. Clearly, we need to take into the account the probability of energy arrival when determining the optimal transmit power policy for each SBS. The amount of available energy at the transmitter will vary according the amount of power  transmitted and the energy arrival during each transmission interval. Naturally,  the competition above can be modeled and analyzed by game theoretic tools. However, unlike in traditional power control games, the actions and the payoffs of the transmitters at successive transmission intervals are correlated. This correlation is taken into account in the proposed stochastic game model. For this game model,  the Nash equilibrium power control policy is obtained as the solution of a quadratic programming problem. For the case when the number of SBSs is very large (i.e., an ultra-dense small cell network (SCN) \cite{5g}), the stochastic game is approximated by a mean field game (MFG). In general, MFGs are designed to study the strategic decision making in very large populations of  interacting individuals. Recently, in \cite{prab2}, the authors proposed an MFG model to determine the optimal power control policy for a finite battery powered SCN with no energy replacements. However, the stochastic nature of energy arrival for small cells was not considered. In this paper, we will consider the case where the battery can be recharged using random energy arrivals. By solving a set of forward and backward partial differential equations, we derive a distributed power control policy for each SBS using a stochastic MFG model.

The contributions of the paper can be summarized as follows.
\begin{enumerate}
\item For a two-tier macrocell-small cell network, we consider a centralized energy harvesting mechanism for the SBSs in which energy is harvested and then distributed to the SBSs through a CES. Unlike in \cite{nicolo}, the CES can have any finite number of energy levels in its storage, not only 0 and 1. Note that the concept of CES is somewhat similar to the concept of dedicated power beacons  for wireless energy transfer to users in cellular networks~\cite{khuang14,wpcn}. 
Moreover, in a cloud-RAN architecture~\cite{cran}, where along with data processing resources, a centralized cloud can also act as an energy farm that distributes energy to the remote radio heads each of which acts as an SBS. Note that the SBSs are not restricted to indoor deployments.
Subsequently, we formulate the power control problem for the MBS and SBSs as a discrete single-controller stochastic game with two players. Also, in this paper, we use the signal-to-interference-plus-noise-ratio (SINR) model instead of an SNR model which has been commonly used in other research works on energy harvesting communication. Consideration of random energy arrivals along with both co-tier and cross-tier interferences in the downlink power control problem is the major novelty of the paper.

\item The existence of the Nash equilibrium and pure stationary strategies for this single-controller stochastic game is proven. The power control policy is derived as the solution of a quadratic-constrained quadratic programming problem. 

\item When the network becomes very dense, a stochastic MFG model is used to obtain the power control policy  as a solution of the forward and backward differential equations. In this case, each SBS can harvest, store energy and transmit data by itself.  

\item An algorithm using finite difference method is proposed to solve these forward-backward differential equations for the MFG model.

\end{enumerate}
Numerical results demonstrate that the proposed  power control policies offer reduced outage probability for the users served by the SBSs when compared to the power control policies using a simple Stackelberg game wherein each SBS tries to obtain the target ${\rm SINR}$ of its users without considering the distribution of energy arrivals. 

The rest of the paper is organized as follows. Section II describes the system model and assumptions. The formulation of the single-controller stochastic game model for multiple SBSs is presented in Section~III. In Section~IV, we derive the distributed power control policy using a MFG model when the number of SBSs increases asymptotically. Performance evaluation results are presented in Section~V before the paper is concluded in Section VI.

\section{System Model and Assumptions}


\subsection{Energy Harvesting Model}
We consider a single macrocell  overlaid with $M$ small cells. The downlink co-channel time-slotted transmission of the MBS and SBSs is considered and it is assumed that each BS serves only a single user on a given transmission channel during a transmission interval (e.g., time slot). 
The MBS uses a conventional power source and its transmit power level is quantized~\cite{prab} into a discrete set of power levels $\mathcal{P} = \{p_0^{min},...,p_0^{max}\}$, where the subscript $0$ denotes the MBS.  On the other hand,  the SBSs receive energy from a centralized energy storage (CES),  which harvests renewable energies from the environment. We assume that only the CES can store energy for future use and each SBS must consume all the energy it receives from the CES at every time slot. The energy 
arrives at the CES in the form of packets (one energy packet corresponds to one energy level in CES). The quantization of energy arrival was assumed in other research studies such as in \cite{blasco}. The number of energy packet arrivals $\varphi(t)$ during any time interval $t$ is discrete and follows an  arbitrary distribution, i.e., $\mathrm{Pr}(\varphi(t) = X)$.  We assume that the battery at the CES has a finite storage $S$. Therefore, the number of energy packet arrivals is constrained by this limit and all the exceeding energy packets will be lost, i.e., $\mathrm{Pr}(\varphi(t)=S) =\mathrm{Pr}(\varphi(t) \geq S)$. The statistics of energy arrival is known {\em a priori} at both the MBS and the CES. At time $t$, given the battery level $E(t)$, the number of energy packet  arrivals $\varphi(t)$, and the energy packets $Q(t)$ that the CES distributes to the $M$ SBSs, the  battery level $E(t+1)$ at the next time slot can be calculated as follows:
 		 \begin{equation}\label{energy}
 		 		E(t+1) = E(t) - Q(t)+ \varphi(t).
 		 \end{equation} 

Given $Q(t)$ energy packets to distribute, the CES will choose the best allocation method for the $M$ SBSs according to their desired objectives. Denoting the slot duration as $\Delta T$ and the volume of one energy packet as $K$, we have the energies distributed to the $M$ SBSs at time $t$ as $(p_1(t)\Delta T,p_2(t) \Delta T,\cdots,p_M(t) \Delta T)$ where $p_i(t)$ is the transmit power of SBS $i$ at time $t$. Clearly we must have:

\begin{equation}\label{power}
\sum^{M}_{i=1} p_i(t)= \frac{K}{\Delta T} Q(t).
\end{equation}



From the causality constraint, $E(t) \geq Q(t) \geq 0$, i.e., the CES cannot send more energy than that it currently possesses.  Note that $E(t)$ is the current battery level which is an integer and has its maximum size limited by $S$. Since the battery level of the CES and the number of packet arrivals are integer values, it follows from (\ref{energy}) that $Q(t)$ is also an integer.

Similar to \cite{harpreet} and \cite{guar}, in our system model, the conflict between the CES and the MBS arises due to the interferences between the MBS and the SBSs. Clearly, if the MBS transmits with large power to achieve the ${\rm SINR}$ targets for macrocell users, it will cause high interference to the small cell users. This means, the SBSs will need to transmit with larger power to combat this  cross-tier interference. The CES and the MBS have different objective functions and are free to choose any actions  that maximize their own objectives (i.e., non-cooperative game). Since the SBSs can only use renewable energy,  it is crucial for them to use their harvested energy economically. Also, unlike a  traditional one-shot transmit power control game, the CES needs to take into the account the future payoff given the current battery size and the probability of energy arrivals. In summary, for our CES model, at each time slot, we will have a random  battery size at the CES and our objective is to maintain the long-term average ${\rm SINR}$ close to the target value as much as possible and thus improve the outage probability. 

Without a centralized CES-based architecture, each SBS can have different amount of harvested energy and in turn battery levels at each time slot, which will make this problem a multi-agent stochastic game \cite{bowling}. 
Although this kind of game can be heuristically solved by using Q-learning \cite{hu}, the conditions for convergence to a Nash equilibrium are often very strict and in many cases impractical. By introducing the CES, the number of the possible states of the game is simplified into the battery size of the CES, and the multi-player game is converted into a two-player game.  Another benefit of the centralized CES-based architecture is that the energy can be distributed based on the  channel conditions of the users served by the SBSs so that the total payoff will be higher than the case where each SBS individually stores and consumes the energy. For ease of exposition, in this paper, we consider an ideal energy transfer from the CES to the SBSs. However, to model a simple energy loss, we can add a fixed percentage of loss into the energy consumption of the CES at each time slot.


 All the symbols that are used in the system model and Section III are listed in \textbf{Table I}.

\begin{table*}[ht]
\scriptsize
\caption{List of symbols used for the single-controller stochastic game model}
\centering
\begin{tabular}
{|c|l|c|l|}
\hline 	
\hline
$\bar{g}_i$      & Average channel gain between BS $i$ and its associated user
& 
$\bar{g}_{i,j}$  & Average channel gain between BS $j$ and user of BS $i$
\\
$\lambda_0$ ($\lambda_1$)   & Target SINR for MBS (SBS) &
$E(t)$		& (Discrete) Battery level of CES at time $t$ 
\\
$\Delta T$& Duration of one time slot in seconds &
$Q(t)$ & Number of quanta distributed by the CES at time $t$ \\
$\bar{I}_0(t)$ & Average interference at the user served by the MBS at time $t$&
$\bar{I}_i(t)$ & Average interference at the user served by SBS $i$ at time $t$\\
$S$ & Maximum battery level of the CES &
$\mathcal{P}$ & Finite set of transmit power of the MBS \\
$\mathbf{m}$, $\mathbf{n}$  & 
\parbox[t]{6cm}{Concatenated mixed-strategy vector for the MBS and the CES, respectively}
&
$\mathbf{m}(s)$, $\mathbf{n}(s)$  & 
\parbox[t]{6cm}
{Probability mass function for actions of the MBS and the CES, respectively, when $E(t)=s$} \\
$\mathbf{m}(s,p)$  & 
{Probability that the MBS chooses power $p\in \mathcal{P}$ when $E(t)=s$}
&
$\mathbf{n}(s,i)$  & Probability that the CES sends $i$ quanta when $E(t)=s$ \\
$\varphi(t)$ & Energy harvested at time $t$ 
& $\pi_s$ & Probability that the CES starts with battery level $s$
\\
$\beta$	& Discount factor of the stochastic game &
$U_0,U_1$	& Utility function of the MBS and the CES, respectively 
\\
$R_0,R_1$	& Payoff matrix for the MBS and the CES, respectively &
$\phi_0$, $\phi_1$  & 
\parbox[t]{6cm}
{Discounted sum of the value function of the MBS and the CES, respectively}
\\
\hline
\end{tabular}
\end{table*}

\subsection{Channel and SINR Model}

The received ${\rm SINR}$ at the user served by SBS $i$ in the downlink at time slot $t$ is defined as follows:
 \begin{equation}
	 \gamma_i(t) = \frac{p_i(t) g_{i,i}}{I_i(t)},
 \end{equation}
where $I_i(t) = \sum\limits^{M}_{j \neq i}p_jg_{i,j} + p_0 g_{i,0} $ is the interference caused by other BSs. $g_{i,0}$ is the channel gain between MBS and the user served by SBS  $i$,
$g_{i,i}$ represents the channel gain between SBS $i$ and the user it serves, and $g_{i,j}$ is the channel gain between SBS $j$ and the user served by SBS $i$. Finally, $p_i(t)$ represents the transmit power of SBS $i$ at time $t$.  The transmit power of MBS $p_0(t)$ belongs to a discrete set $\{p_0^{min},...,p_0^{max}\}$~\cite{prab}. We ignore the thermal noise assuming that it is very small compared to the cross-tier interference. 

Similarly, the SINR at a macrocell user can be calculated as follows:
 \begin{equation}\label{MUsinr}
 	 \gamma_0(t) = \frac{p_0(t) g_{0,0}}{ I_0(t)+N_0},
 \end{equation}
where $I_0(t) = \sum\limits^{M}_{i=1}p_ig_{0,i}$ is the cross-tier interference from $M$ SBSs to the macrocell user, $g_{0,0}$ denotes the channel gain between the MBS  and its user, $g_{0,i}$ represents the channel gain between SBS $i$ and macrocell user, and $N_0$ is the thermal noise.

The channel gain $g_{i,j}$ is calculated based on  path-loss and fading gain as follows:
\begin{equation}
	g_{i,j} = |h|^2r_{i,j}^{-\alpha},
\end{equation}
where $r_{i,j}$ is the distance from BS $j$ to user served by BS $i$, $h$ follows a Rayleigh distribution, and $\alpha$ is the path-loss exponent. We assume that the $M$ SBSs are randomly located around the MBS and the users are uniformly distributed  within their coverage radii $r$. During a transmission interval (i.e., a time slot), only one user is served by each SBS in the downlink direction.

\section{Formulation and Analysis of the Single-Controller Stochastic Game}

A stochastic game  is a multiple stage game where it can have different states at each stage.  Each player chooses an action from a finite set of possible actions (which can be different at each stage). The players'
actions and the current state jointly determine the payoff to each player and the transition
probabilities to the succeeding state. The total payoff to a player is defined as the discounted sum of the stage payoffs or the limit inferior of the averages of the stage payoffs. The transition of the game at each time instant follows Markovian property, i.e., the current stage only depends on the previous one.  

\subsection{Utility Functions of MBS and SBSs}
In our model, the MBS and the SBSs try to maintain the average ${\rm SINR}$ of their users to be close to some targets. Note that a large target ${\rm SINR}$ means that a high transmit power will be required for the SBS which could be limited by the energy arrival rate and the battery size of the CES. Also, a higher transmit power means a higher level of  interference to other users. Similar to \cite{harpreet} and \cite{duy}, the utility function of the MBS at time $t$ is defined as:
\begin{equation}\label{u0}
\begin{split}
 U_0(p_0,Q,t)= -(p_0(t)\bar{g}_0-\lambda_0 (\bar{I}_0(t)+N_0))^2,
\end{split}
\end{equation}
where $\bar{I}_0(t)=\sum\limits^{M}_{i=1}p_i(t)\bar{g}_{0,i}$ is the average interference at the macrocell user at time $t$, and $\lambda_0$ is the target ${\rm SINR}$ for the macrocell user. Clearly, this utility function is maximized when the ${\rm SINR}$ at the MBS is $\lambda_0$.  If the ${\rm SINR}$ is larger than the target ${\rm SINR}$, this implies that MBS transmits with a larger power than necessary and thus wastes energy. On the other hand, if the ${\rm SINR}$ is smaller than the target ${\rm SINR}$, it implies that energy is not utilized effectively provided it has sufficient energy.

Similarly, the utility function of the CES is defined as follows:
\begin{equation}\label{u1}
	U_1(p_0,Q,t) = -\frac{1}{M}\sum\limits_{i=1}^{M}(p_i(t)\bar{g}_i- \lambda_1 \bar{I}_i(t))^2,
\end{equation}
where $\bar{I}_i(t)=	\sum\limits^{M}_{j \neq i}p_j(t)\bar{g}_{i,j} +p_0(t)\bar{g}_{i,0}$ is the average interference at the user served by SBS $i$ at time $t$. 
The arguments of both  the utility functions demonstrate that the action at time $t$ for the MBS is its {\em transmission power} $p_0(t)$ while the action of the CES is the {\em number of energy packets $Q(t)$ that is used  to transmit data from the SBSs}. Later, in {\bf Remark 2}, we will show that the interference and the transmit power of each SBS can be derived from $Q$ and $p_0$.
The conflict in the payoffs of both the players arises from their transmit powers that directly impact the cross-tier interference. 

Note that the proposed single-controller approach can be extended to consider a variety of utility functions (average throughput, total network throughput, energy efficiency, etc.)


\subsection{Formulation of the Game Model}
Unlike a traditional power control problem, the action space of the CES changes at each time slot and is limited by its battery size. Given the distribution of energy arrival  and the discount factor  $\beta$, the power control problem can be modeled by a single-controller discounted stochastic game as follows:
\begin{itemize}
\item There are two players: one MBS and one CES.
\item The state of the game is the battery level of the CES, which belongs to $\{0,...,S\}$. 
\item At time $t$ and state $s$, the action $p_0(t)$  of the MBS is its transmission power and belongs to the finite set  $\mathcal{P} = \{p_0^{min},...,p_0^{max}\}$. On the other hand, the action of the CES is $Q(t)$, which is the number of energy packets distributed to $M$ SBSs. $Q(t)$ belongs to the set $\{0,...,s\}$.

\item Let $\mathbf{m}$ and $\mathbf{n}$ denote the concatenated mixed-stationary-strategy vectors of the MBS and the CES, respectively. The vector $\mathbf{m}$ is constructed by concatenating $S+1$ sub-vectors into one big vector as $
				\mathbf{m}= \left[\mathbf{m}(0), \mathbf{m}(1), ..., \mathbf{m}(S)\right],$
		 in which each $\mathbf{m}(s)$ is a vector of probability  mass function for the actions of the MBS at state $s$. For example,
		  if the game is in state $s$, $\mathbf{m}(s,p)$ gives the probability that the MBS transmits with power $p$. Therefore, the full form of $\mathbf{m}$ will include the state $s$ and power $p$. However, to make the formulas simple, in the later parts of the paper, we will use $\mathbf{m}$ or $\mathbf{m}(s)$ to denote, respectively, the entire vector or a sub-vector at state $s$, respectively. 

\item Similarly, for the CES, $\mathbf{n}(s,i)$ gives the probability that the CES distributes $i$ energy packets. Note that the available actions of the CES  dynamically vary at each state whereas the available actions for the MBS remain unchanged at every state. 

\item {\bf Payoffs:} At state $s$, if the MBS transmits with power $p_0$ and the CES distributes $Q$ energy packets, the payoff function for the MBS is $U_0(p_0,Q)$ while the payoff function for the CES is $U_1(p_0,Q)$. We omit $t$ since $t$ does not directly appear in $U_1$ and $U_0$.  

\item {\bf Discounted Payoffs:} Denoting by $\beta$ the discount factor ($\beta <1$), the discounted sum of payoffs of the MBS is given as:
\begin{equation}
	\phi_0(s,\mathbf{m},\mathbf{n}) =\lim\limits_{T \rightarrow \infty} \sum\limits^{T}_{t=1} \beta^t \mathbb{E}[U_0(\mathbf{m},\mathbf{n},t)],
\end{equation}
where $\mathbb{E}[U_0(\mathbf{m},\mathbf{n},t)]$ is the average utility of macrocell user  at time $t$ if the MBS and the CES are using strategy $\mathbf{m}$ and $\mathbf{n}$, respectively. Similarly, we define the discounted sum of payoffs $\phi_1$ at the CES. In \cite[Chapter~2]{filar3}, it was proven that the limit of $\phi_0$ and $\phi_1$ always exist when $T \rightarrow \infty$. 

\item {\bf Objective:} To find a pair of strategies $(\mathbf{m}^*,\mathbf{n}^*)$ such that $\phi_0$ and $\phi_1$ become a Nash equilibrium, i.e., $
	\phi_0(s,\mathbf{m}^*,\mathbf{n}^*) \geq \phi_0(s,\mathbf{m}^*,\mathbf{n}) ~~~\forall \mathbf{n} \in \mathcal{N} ~~~
	\text{and}~~~ \phi_1(s,\mathbf{m}^*,\mathbf{n}^*) \geq \phi_1(s,\mathbf{m},\mathbf{n}^*) ~~~\forall \mathbf{m} \in \mathcal{M}
$,
where $\mathcal{M}$ and $\mathcal{N}$ are the sets of strategies of MBS and CES, respectively.
\end{itemize}	
Given the distribution of energy arrival at the CES, the transition probability of the system from state $s$ to state $s'$ under action $Q$ ($0 \leq Q \leq s$) of the CES is given as follows:
	\begin{align}\label{markovState}
					q(s'|s,Q) =
					\begin{cases}
						\mathrm{Pr}(\varphi=s'-(s-Q)), ~~~ &\text{if $s' < S$} \\
						 1-\sum\limits_{X=0}^{S-s}\mathrm{Pr}(\varphi = X) ,  ~~~&\text{otherwise.}
					\end{cases} 
				\end{align} 
Also, we assume that  information about the average channel gains are available to all players. This implies that the single-controller stochastic game we present here will be a perfect information non-cooperative game.

The states of the game can be described by a Markov chain for which the transition probabilities are defined by (\ref{markovState}). Clearly, the CES controls the state of the game while the MBS has no direct influence. Therefore, the single-controller stochastic game can be applied to derive the Nash equilibrium strategies for both the MBS and the CES.

 The two main steps to find the Nash equilibrium strategies are: 
\begin{itemize}
\item First, we build the payoff matrices for the MBS and the CES for every state $s$, where $S \geq s \geq 0$. Denote them by $R_0$ and $R_1$, respectively. 
\item Second, using these matrices, we solve a quadratic programming problem to obtain the Nash equilibrium strategies for both the MBS and the CES.
\end{itemize}

\subsection{Calculation of the Payoff Matrices}

To build $R_0$ and $R_1$, we  calculate $U_0$ and $U_1$ for every possible pair $(p_0,Q)$, where $p_0 \in \mathcal{P}$ and $0 \leq s \leq S$. In this regard,  we first derive the average channel gain $\bar{g}_{i,j}$. Second, from the energy consumed $Q$ and transmission power $p_0$ of the CES and the MBS, 
\begin{figure}
\centering
\includegraphics[width=2in]{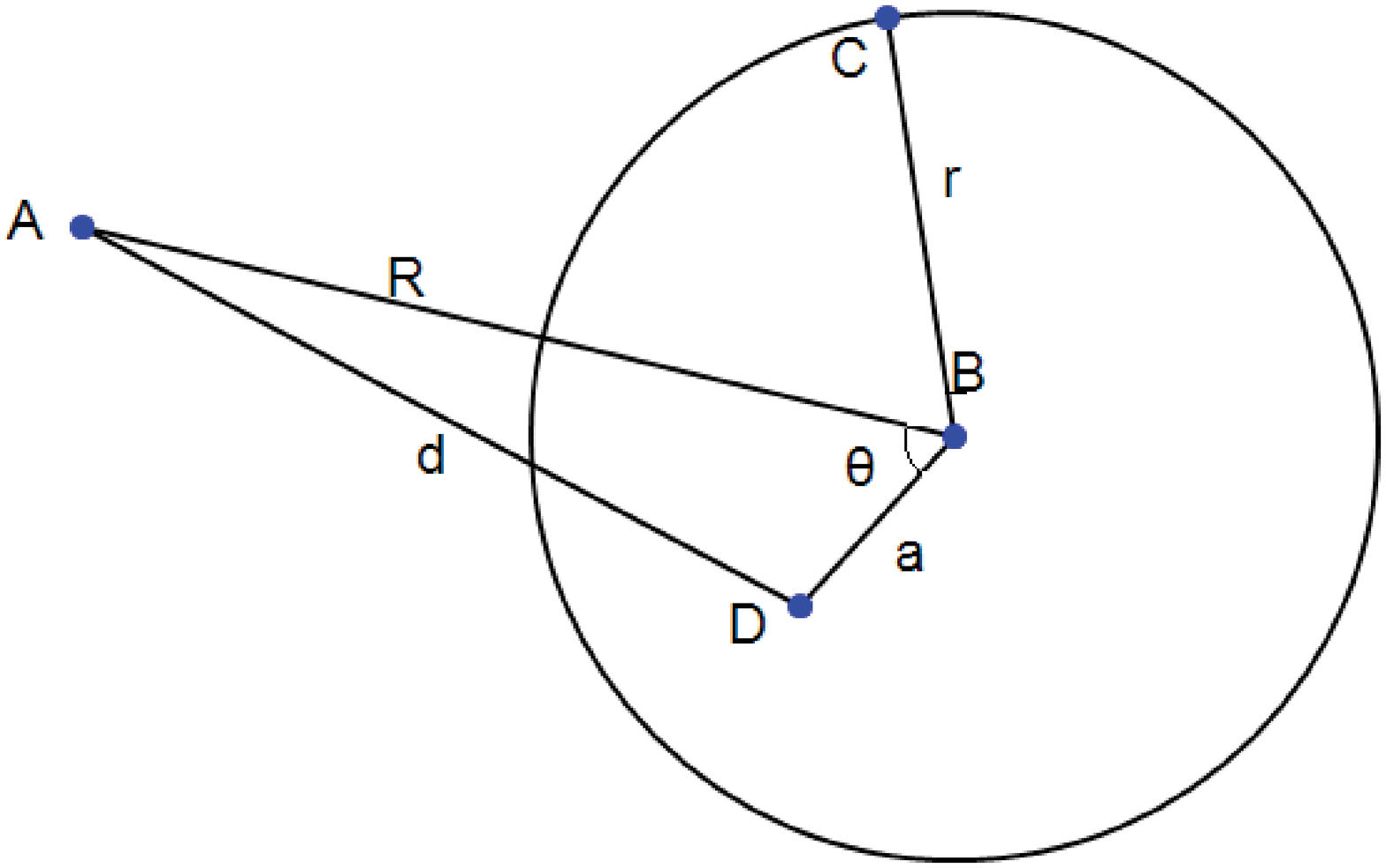}
\caption{Graphical illustration of the two BSs $A,$ $B$, and the user $D$ located within the disk centred at $B$.}
\label{locations}
\end{figure}
respectively, we decide how the CES distributes this energy $Q$ among the SBSs. Then, we calculate the transmit power at each SBS and obtain $U_0$ and $U_1$. The next two remarks provide us with the methods to calculate $U_0$ and $U_1$.
\begin{myrem}\label{rem1}
Given two BSs $A$ and $B$, assume that a user $D$, who is associated with $B$, is uniformly located within the circle centred at $B$ with radius $r$ (Fig.~\ref{locations}). Assume that $A$ does not lie on the circumference of the circle centred at $B$ and $\alpha= 4$. Denote $AB=R$ and $AD=d$, then the expected value of $d^{-4}$, i.e., $\mathbb{E}[d^{-4}]$ is 
$\frac{1}{(R^2-r^2)^2}$. If $A \equiv B$, then  $\mathbb{E}[d^{-4}] = \frac{1-r^{-2}}{r^2}$ given that $r \geq BD \geq 1$.  For other values of $\alpha$, $\mathbb{E}[r_{ij}^{-\alpha}]$ can be easily computed numerically using tools such as \texttt{MATHEMATICA}.
\end{myrem}

\begin{proof}
See {\bf Appendix A.}
\end{proof}
Recalling that $g_{i,j} = |h|^2 r_{ij}^{-4}$ and  that the fading and path-loss are independent, we have $\bar{g}_{i,j} = \mathbb{E}[h^2]\mathbb{E}[r_{ij}^{-4}]$, where $\mathbb{E}[h^2] =\lambda$, if $h$  follows Rayleigh distribution with scale parameter $\lambda$.
Next, we need to find how the CES distributes its energy to each SBS such that $U_1$ is maximized. 

\begin{myrem} (Optimal energy distribution at the CES)\label{rem2}
If at time $t$ the CES distributes $Q$ energy packets to $M$ SBSs and the MBS transmits with power $p_0$, then the transmit powers $(p_1,p_2,...,p_M)$ at the $M$ SBSs are the solutions of the following optimization problem ($t$ is omitted for brevity):	
\begin{equation}\label{dist}	
	\begin{split}	\max\limits_{p_1,p_2,...,p_M}&-\frac{1}{M}\sum\limits_{i=1}^{M}\left(p_i\bar{g}_i- \lambda_1(\sum\limits^{M}_{j \neq i}p_j\bar{g}_{ij} +p_0\bar{g}_{i,0})\right)^2, \\
		\text{s.t.}~~ & \sum\limits^{M}_{i=1} p_i = \frac{K}{\Delta T} Q, \\
		& P_{max} \geq p_i \geq 0, ~~~~ \forall i = 1,2, ..., M,		
	\end{split}
	\end{equation}
where $P_{max}$ is the maximum transmit power of each SBS. Since this problem is strictly concave, the solution $(p_1, ..., p_M)$ always exists and is unique for each pair ($Q,p_0$). Thus, for each pair $(Q,p_0)$, where $Q \in \{0,...,S\}$ and $p_0 \in \{p_0^{min},...,p_0^{max}\}$, we have unique values for $U_0(p_0,Q)$ and $U_1(p_0,Q)$. 
\end{myrem}

Based on the remarks above, for each combination of $Q$ and $p_0$, we can find the unique payoff $U_0$ and $U_1$ of MBS and CES. Since $Q$ and $p_0$ belongs to discrete sets we can find the payoff for all of the possible combinations between them. Thus, we can build the pay-off matrix $R_0$ for the MBS and $R_1$ for the CES. The matrix $R_0$ has the form of a block-diagonal matrix $\mathrm{diag}(R_0^0,...,R_0^S)$, where each sub-matrix $R_0^s = (U_0(p_0,j))^{\mathcal{P}\times\{0,...,s\}}$, with $p_0 \in \mathcal{P}$  and $j \in \{0,...,s\}$ is the matrix of all possible payoffs for the MBS at state $s$. Similarly, we can build $R_1$, which is the payoff matrix for the CES. A detailed explanation on how we use them will be given in the next subsection.

\subsection{Derivation of the Nash Equilibrium}

If we know the strategy $\mathbf{m}_0$ of the MBS, the discount factor $\beta$, and the probability $\pi_s$ that the CES starts with $s$ energy packets in the battery, then the stochastic game is reduced to a simple MDP problem with only one player, the CES. For this case, denote the CES's best response strategy to $\mathbf{m}_0$ by $\mathbf{n}$. Then the CES's value function $\phi_1(s,\mathbf{m}_0,\mathbf{n})$, where $s=0, ..., S$, is the solution of the following MDP problem \cite[Chapter~ 2]{filar3}:		  
\begin{equation}\label{P}
 \begin{split}
	&\min\limits_{\phi_1} \sum\limits_{s=0}^{S}\pi_s \phi_1(s,\mathbf{m}_0,\mathbf{n}),\quad\quad \\
	\mbox{s.t.}\quad \phi_1(s,\mathbf{m}_0,&\mathbf{n}) \geq r_1(s,\mathbf{m}_0,j) + \beta \sum\limits_{s'=0}^{S}q(s'|s,j)\phi_1(s',\mathbf{m}_0,\mathbf{n}), \\ & \forall s,j,~~~ 0 \leq j \leq s~~ \text{and} ~~ 0 \leq s \leq S, 					  
			  \end{split}
 \end{equation}	
with $r_1(s,\mathbf{m}_0,j) = \sum_{p_0 \in \mathcal{P}}  U_1(p_0,j)\mathbf{m}_0(s,p_0)$ is the average payoff for the CES at state $s$ when it consumes $j$ quanta of energy. Using the Dirac function $\delta$, the dual problem can be expressed as
		\begin{equation}\label{D}
			\begin{split}
				\max\limits_{x} &\sum\limits_{s=0}^{S}\sum\limits_{j=0}^{s} r_1(s,\mathbf{m}_0,j)x_{s,j},\\
			\mbox{s.t.}~~	\sum\limits_{s=0}^{S} \sum\limits_{j=0}^{s}[\delta(s-s')-
				\beta& q(s'|s,j)]x_{s,j} = \pi_{s'}, ~~~ \forall 0 \leq s' \leq S, \\
				x_{s,j} \geq 0 ~~ \forall s,j,&~~~ 0 \leq j \leq s~~ \text{and} ~~ 0 \leq s \leq S, 
			\end{split}
		\end{equation}
where $\delta(s)=1$ if $s=0$ and $\delta(s)=0$, otherwise. 

By solving the pair of linear programs above, the probability that the SBS chooses action $j$ at state $s$ can be found as 
$
\mathbf{n}(s,j) = \frac{x_{s,j}}{ \sum_{j=0}^{s} x_{s,j}}.
$ Using some algebraic manipulations, we can convert the optimization problem in (\ref{P}) into a matrix form as: 		 
\begin{align*}
\min\limits_{\mathbf{\phi}_1} &~~ \mathbf{\pi}^{\text{T}}\mathbf{\phi}_1, \\
\text{s.t.} & ~~H\mathbf{\phi}_1 \geq R_1^{\text{T}}\mathbf{m}_0 \tag{P},
\end{align*}
and its dual as 		 
\begin{align*}
\max\limits_{\mathbf{x}} ~~ &\mathbf{m}_0^{\text{T}}R_1\mathbf{x},  \\
 \text{s.t.}  ~~&
\mathbf{x}^{\text{T}} H = \pi^{\text{T}}, \\
&\mathbf{x} \geq 0,		 		 \tag{D}
\end{align*}	 
where $R_1$ is the payoff matrices of the CES. Combining the primal and dual linear programs (i.e.,  (P) and ({D}) above) and using the same notations, we have the following theorems.
		
\begin{mythe}[Nash equilibrium strategies~\cite{Filar1}]  If the state space and the action space are finite and discrete, and the transition probabilities are controlled only by player 2 (i.e., the CES), then there always exists a Nash equilibrium point $(\mathbf{m},\mathbf{n})$ for this stochastic game. Moreover, a pair $(\mathbf{m},\mathbf{n})$  is a Nash equilibrium point of a general-sum single-controller discounted stochastic game if and only if it is an optimal solution of a (bilinear) quadratic program given by  		
			\begin{equation} \label{quadr}
			\begin{split}
					\max\limits_{\mathbf{m},\mathbf{x},\mathbf{\phi},\mathbf{\xi}} ~~ &[\mathbf{m}(R_0+R_1)\mathbf{x}-\mathbf{\pi}^{\text{T}} \mathbf{\phi}_1 - \mathbf{1}^{\text{T}}\mathbf{\xi}], \\
					\text {s.t. }~~ &  H\mathbf{\phi}_1 \geq R_1^{\text{T}}\mathbf{m}, \\
					& \mathbf{x}^{\text{T}}H = \mathbf{\pi}^{\text{T}}, \\				
					& R_0^s\mathbf{x}(s) \leq \xi_s \mathbf{1}, ~~~\forall s = 0,..., S,\\
					& \mathbf{m}(s)^{\text{T}}\pmb{1} = 1, ~~~\forall s = 0,..., S,\\			 
					& \mathbf{m,x} \geq \mathbf{0},	
					\end{split}		
				\end{equation}
		where $\xi_s$ is the maximum average payoff of the MBS at state $s$. The sub-vector strategy $\mathbf{n}(s)$ of the CES at state $s$ is calculated from $\mathbf{x}$ as: 
		\begin{align}
			\mathbf{n}(s) =
				\frac{\mathbf{x}(s)}{ \mathbf{x}(s)^{\text{T}}\mathbf{1}}. 	
		\end{align}		
	\end{mythe}
		
We can define different utility functions for the MBS and the SBS and apply the same method to achieve the Nash equilibrium. As long as the number of states is finite and the transition probabilities and the payoff matrices are known, a Nash equilibrium point always exists.

\begin{mythe} [Best response strategy for the MBS] \label{mbs} Given a stationary strategy $\mathbf{n}$ of the CES, there exists a pure stationary strategy $\mathbf{m}$ as the best response for the MBS. Similarly, for any stationary strategy $\mathbf{m}$ of the MBS, there exists a pure stationary best response $\mathbf{n}$ of the CES.
\end{mythe}

\begin{proof}
 See \textbf{Appendix B}. 
 \end{proof}

Because for every mixed strategy of the CES, the MBS can find a pure stationary strategy as a best response, we only need to find Nash equilibrium where the strategy of the MBS is deterministic. From there, this problem can be converted to a mixed-integer program with $\mathbf{m}$ as a vector of 0 and 1. We can use a brute-force search to obtain an equilibrium point. For each feasible integer value of $\mathbf{m}$ we insert it into (\ref{quadr}) to obtain $\mathbf{n}$. If the objective is zero, then $(\mathbf{m},\mathbf{n})$ is the equilibrium point. This theorem implies that the optimization problem in (\ref{quadr}) can be solved in a finite amount of time. 

Notice that since $R_0+R_1$ is not a positive semi-definite matrix, there can be several solutions for the quadratic programming problem in (\ref{quadr}), i.e., multiple Nash equilibrium. Therefore, to make the chosen Nash equilibrium point more meaningful, we use the following lemma from \cite{Filar1}.  

\begin{mylem}(Necessary and sufficient conditions for the Nash equilibrium)
$\mathbf{m}$ and $\mathbf{n}$ constitute a pair of Nash equilibrium policies for the MBS and the CES if and only if 
\begin{equation}
\mathbf{m}(R_0+R_1)\mathbf{x}-\mathbf{\pi}^{\text{T}} \mathbf{\phi}_1 - \mathbf{1}^{\text{T}}\mathbf{\xi} = 0. \label{equiCondition}
\end{equation}
\end{mylem}

Since $\mathbf{\pi}_s$ is the probability that the CES starts with $s$ energy level in the battery at starting time, from (\ref{P}), $\mathbf{\pi}^{\text{T}}\phi_1$ is the average value function of the CES with respect to the energy arrival rate and the strategies $\mathbf{m},\mathbf{n}$. We are interested in the Nash equilibrium that maximizes the average payoff of the CES (i.e., the SBSs). This bias towards the SBSs is crucial as the available energy of the CES is limited by the randomness of the energy arrival process and thus the SBSs are more likely to suffer  when compared to the MBS.  Using the lemma above, we change the problem in (\ref{quadr}) to a quadratic-constrained quadratic programming (QCQP) as stated below. 

\begin{mypro}(Nash equilibriums that favor the SBSs)
The Nash equilibrium $(\mathbf{m},\mathbf{n})$ that has the best payoff for the CES is a solution of the following QCQP problem:
\begin{equation}\label{QCQP}
			\begin{split}
				&\max\limits_{\mathbf{m},\mathbf{x},\mathbf{\phi},\mathbf{\xi}} \mathbf{\pi}^{\text{T}} \mathbf{\phi}_1, \\
				\text{s.t.}~~~ & \mathbf{m}(R_0+R_1)\mathbf{x}-\mathbf{\pi}^{\text{T}} \mathbf{\phi}_1 - \mathbf{1}^{\text{T}}\mathbf{\xi}=0, \\
				&\text{all constraints from (\ref{quadr}). }
			\end{split}
		\end{equation}
 \end{mypro}
By solving this QCQP (i.e., by using a brute-force search), we obtain a Nash equilibrium that returns the best average payoff for the CES.  Again, we can still have multiple Nash equilibrium in this case, but all of them must return the same payoff for the CES. Because there may be multiple solutions, the CES and the MBS need to exchange information so that they agree on the same Nash equilibrium.

\begin{algorithm}
 \caption{Nash equilibrium for the stochastic game}\label{StocAlgo}
     \begin{algorithmic}[1]     	  	
     	  		\State {The MBS and the CES build their reward matrices $R_0$ and $R_1$. For each possible pair of energy level and transmit power $(Q,p_0)$, the CES solves (\ref{power}) to obtain a unique tuple $(p_1,p_2,...,p_M)$ and record these results.}
     	  	  	\State The MBS and the CES calculate their strategy $\mathbf{m}$ and $\mathbf{n}$, respectively, by solving (\ref{QCQP}).
     	  	  	\State At time $t$, the CES sends its current battery level $s$ to  the MBS. It also randomly chooses an action $Q$ using the probability vector $\mathbf{n}(s)$. 
     	  	  	\State The MBS then randomly picks  power $p_0$ using distribution $\mathbf{m}(s)$ and sends it back to the CES. Based on $p_0$ and $Q$, the CES searches its records and retrieves the corresponding tuple $(p_1,...,p_M)$.  
     	  	  	\State The CES distributes energy $(p_1\Delta T,p_2\Delta T,...,p_M\Delta T)$, respectively, to the $M$ SBSs.	  	  	  	  	  		  	  	
     	  	\end{algorithmic}
     	\end{algorithm}
From \textbf{Theorem 2}, we know that there exists an equilibrium with pure stationary  strategies for both the MBS and CES. Recall that with pure strategy, the action of each player is a function of the state. Thus, if we can obtain this equilibrium, the CES can predict which transmit power $p_0$ the MBS will use based on the current state, without exchanging information with the MBS and vice versa.

\subsection{Implementation of the Discrete Stochastic Game}

For the discrete stochastic control game with the CES,  each SBS first needs to send its location and average fading channel information $\mathbb{E}[h^2]$ of its user to the CES. Then the CES and the MBS will exchange information so that the MBS can have complete knowledge about the average channel gains at each SBS. Since we only use average value, the CES and the MBS only need to re-calculate the Nash equilibrium strategies when either the locations of SBSs change, e.g., some SBSs go off and some are turned on, or when the average of channel fading gain $h$ changes, or when distribution of energy arrival $\varphi(t)$ at the CES changes. Note that the SBSs and the MBS only need to send the channel gain information of the corresponding users (to be served) to the CES.

\section{Mean Field Game (MFG) for Large Number of Small Cells} 

The main problem of the two-player single-controller stochastic game is the ``curse of dimensionality".  The time complexity of \textbf{Algorithm~1} increases exponentially with the number of states $S$ or the maximum battery size. Note that $R_0$ and $R_1$ have dimensions of $|\mathcal{P}| \times S(S+1)/2$, so the complexity increases proportionally to $S$. Moreover, unlike other optimization problems,  we are unable to relax the QCQP in (\ref{QCQP}), because \textbf{Theorem~1} states that the Nash equilibrium must be the global solution of the quadratic program in (\ref{quadr}). To tackle these problems, we extend the stochastic game model to an MFG model for a very large number of players.

The main idea of an MFG is the assumption of similarity, i.e., all players are identical and follow the same strategy. They can only be differentiated by their ``state" vectors.  If the number of players is very large, we can assume that the effect of a specific player on other players is nearly negligible. Therefore, in an MFG, a player does not care about others' states but only act according to a ``mean field" $m(t,\mathbf{s})$, which usually is  the probability distribution of state $s$ at time instant $t$ \cite{guent}. In our energy harvesting game,  the state is the battery $E$ and the mean field $m(t,E)$ is the probability distribution of energy in the area we are considering. When the number of players $M$ is very large, we can assume that $m(t,E)$ is a smooth continuous distribution function. We will express the average interference at an SBS as a function of the mean field $m$. 


 All the symbols used in this section are listed in \textbf{Table II}.

\begin{table*}
\caption{List of symbols used for the MFG game model}
\centering
\scriptsize
\begin{tabular}{|c|l|c|l|}
\hline \hline	
$\bar{g}$          & Average channel gain from a generic SBS to another user & $p(t,R)$	          & Transmit power at a generic SBS as a function of $R$
 \\
$E$		          & (Continuous) Battery level of an SBS & $m(t,E)$           & Probability distribution of energy $E$ at time $t$\\
$R$		          & Energy coefficient $e^R = E$ &$m(t,R)$           & Probability distribution of energy coefficient $R$ at time $t$\\
$W_t$              & Wiener process at time $t$& $\bar{p}(t)$       & Average transmit power of a SBS at time $t$\\
$p(t,E)$  	       & Transmit power at a generic SBS as a function of $E$ (or $R$)& $\sigma$           & Intensity of energy arrival or loss\\
\hline \hline
\end{tabular}
\end{table*}

\subsection{Formulation of the MFG}

Denote by $E(v)$ the available energy in the battery of an SBS at time $v$. Given the transmission strategies of other SBSs, each SBS will try to maximize its long-term generic utility function by solving the following optimal control problem:
\begin{align}
\min_{p}& ~ U(0,E(0)) = \mathbb{E}\left[\int^{T}_{0}(p(v,E(v))g-\lambda( I(v)+N_0) )^2dv\right],\\
\text{s.t.}~~~ & dE(v) = -p(v,E(v))dv + \sigma dW_v, \label{mfg:1} \\
&  E(v) \geq 0 ,~~  p(v) \geq 0,
\end{align}
where $I(v)$ is the generic interference at a user served by an SBS at time $v$ and $g$ is the channel gain between a generic SBS and its user. 
The mean field $m(v,E)$ is the probability distribution of energy $E$ in the area at time $v$. Using $M$ as the number of SBSs in a macrocell and assuming that the other SBSs have the same average channel gain $\bar{g}$ to the user of the current generic SBS,  the average interference $I(v)$ at the user served by a generic SBS can be expressed as
$
		I(v) = M \bar{g}\bar{p}(v),
$
where $\bar{p}(v) = \int_{0}^{\infty} p(v,E)m(v,E)dE$ can be understood as the average transmit power of ``another" generic SBS. Since the MFG assumes similarity, $\bar{p}(v)$ can be considered as the average transmit power of a generic SBS at time $v$. To make the notation simpler, we denote $\bar{\lambda} = \lambda\bar{g}M$.

Thanks to similarity, all the SBSs have the same set of equations and constraints, so the optimal control problem for the $M$ SBSs reduces to finding the optimal policy for only one generic SBS.
Mathematically,  if an SBS has infinite available energy, i.e., $E(0)= \infty$, it will act as an MBS.  However, for simplicity, we will assume that only the SBSs are involved in the game and the interference from the MBS is constant, which is included in the noise $N_0$ as in \cite{ali}. Except that, the system model and the optimization problem here are similar to those in the discrete stochastic game model. 

Assuming that the SBSs are uniformly distributed within the macrocell with radius $r$ centred at the MBS, the average interference from the MBS to a generic user served by an SBS can be easily derived by using  a method similar to that described in \textbf{Remark~\ref{rem1}}. For the MFG model, the energy level $E$ is a continuous non-negative variable.
The equality in (\ref{mfg:1}) shows the evolution of the battery, where $\sigma$ is a constant which is proportional to the maximum energy arrival during a time interval. $W_v$ is a Wiener process, thus $dW_v = \epsilon_v dv$, where $\epsilon_v$ is a Gaussian random variable with mean zero and variance 1. This model of evolution for battery energy was mentioned in $\cite{tembine}$. The inflexibility of the energy arrival is the main disadvantage of using the MFG model compared to the discrete stochastic game model. The random arrival of energy is configured as ``noise", so this can be either positive or negative. We can consider the negative part as the battery leakage and internal energy consumption. The final inequalities are the causality constraints: the battery state $E(v)$ and transmit power must always be non-negative. To guarantee this positivity we follow \cite{MFGlab} and change the energy variable $E(v)$ to
$
E(v) = e^{R(v)}.
$
This conversion is a bijection map from $E(v)$  to $R(v)$, thus we can write
$m(v,E) = m(v,R)$ and $p(v,E)=p(v,R)$, where $\infty > R > -\infty$. The new optimal control problem can be rewritten as
\begin{align}\label{mf}
 \min_{p(.)}~ U(0,R(0)) =  & \nonumber \\
 \mathbb{E}\left[\int^{T}_{0}(p(v,R(v))g-\bar{\lambda}\bar{p}(v)-\lambda N_0)^2dv\right], & \\
\text{s.t.}~~~ dR(v) = -p(v,R(v))e^{-R(v)}dv + \sigma e^{-R(v)} dW_v, 	&\label{mfg:2} \\
p(v) \geq 0.&
\end{align}
To obtain the power control policy $p$,  first, we derive the forward-backward differential equations from the above problem.
Then, we apply finite difference method to numerically solve these equations.

\subsection{Forward-Backward Equations of MFG}

Assuming that the optimal control above starts at time $t$ with $T \geq t \geq 0$, we obtain the Bellman function $U(t,R)$ as
\begin{equation}
	U(t,R(t)) = \mathbb{E}\left[\int^{T}_{t}(p(v,R(v))g-\bar{\lambda} \bar{p}(v)-\lambda N_0)^2dv\right].
\end{equation}
From this function, at time $t$, we obtain the following Hamilton-Jacobi-Bellman (HJB) \cite{MFGlab} equation:
\begin{eqnarray}
	 \partial_t U + 
\min\limits_{p\geq 0} \left\{\left(p(t,R)g-\bar{\lambda} \bar{p}(t) -\lambda N_0 \right)^2 - p(t,R)e^{-R}\partial_R U(t,R) \right\} & \nonumber\\
	 + \frac{\sigma^2}{2}e^{-2R} \partial_{RR}^2 U = 0, &
\end{eqnarray}
where $\bar{p}(t) = \int_{-\infty}^{\infty} e^{R} p(t,R)m(t,R)dR$ is the average transmit power at a generic SBS. 
The Hamiltonian $\min\limits_{p\geq 0} \left\{\left(p(t,R)g-\bar{\lambda} \bar{p}(t) -\lambda N_0 \right)^2 - p(t,R)e^{-R}\partial_R U(t,R) \right\}$ is given by the Bellman's principle of optimality. By  applying the first order necessary condition, we obtain the optimal power control  as follows:
\begin{equation}
	 p^*(t,R) = \left[\frac{\bar{\lambda}\bar{p}(t) +\lambda N_0}{g} + \frac{e^{-R}\partial_R U}{2g^2} \right]^+. \label{powerMFG}
\end{equation}


\begin{myrem}
The Bellman $U$, if exists, is a non-increasing function of time and energy. Therefore, we have $\partial_R U \leq 0$ and  $\partial_t U \leq 0$. 
\end{myrem}
 From equation ($\ref{powerMFG}$), given the current interference ${\bar{\lambda}}\bar{p}(t)+\lambda N_0$ at a user, the corresponding SBS  will transmit less power based on the future prospect $\frac{e^{-R}\partial_R U}{2g^2} $. If the future prospect is too small, i.e., $\frac{e^{-R}\partial_R U}{2g^2} < -\frac{\lambda\bar{p}(t)+\lambda N_0}{g}$, it stops transmission to save energy.
  
Replacing $p^*$ back to the HJB equation, we have
\begin{eqnarray}
\partial_t U+ \frac{\sigma^2}{2} e^{-2R} \partial^2_{RR} U + (\bar{\lambda}\bar{p}(t)+\lambda N)^2 - & \nonumber \\ \left(\left[\bar{\lambda}\bar{p}(t)+\lambda N+ \frac{e^{-R}\partial_R U}{2g}\right]^+\right)^2 = 0,&
\end{eqnarray}
  which has a simpler form as follows:
  \begin{equation}
	  \partial_t U+ \frac{\sigma^2}{2}e^{-2R} \partial^2_{RR} U = \left(pg\right)^2- (\bar{\lambda}\bar{p}(t)+\lambda N)^2.
  \end{equation}
  Also, from (\ref{mfg:2}), at time $t$, we have the Fokker-Planck equation \cite{guent} as:
  \begin{equation}\label{fp}
	  \partial_t m(t,R) =  \partial_R (pe^{-R}m) + \frac{\sigma^2}{2} \partial_{RR}(me^{-2R}),
  \end{equation}
where $m(t,R)$ is the probability density function of $R$ at time $t$. 
Combining all these information, we have the following proposition.
  \begin{mypro}\label{PDEs}
	  The value function and the mean field $(U,m)$ of the MFG defined in (\ref{mf}) is the solution of the following partial differential equations:
	  \begin{align}
		  \partial_t U+ \frac{\sigma^2}{2}e^{-2R} \partial^2_{RR} U =& \left(pg\right)^2- (\bar{\lambda}\bar{p}(t)+\lambda N_0)^2, \label{valueF}\\
		  p(t,R) =& \left[\frac{\bar{\lambda}\bar{p}(t)+\lambda N_0}{g} + \frac{e^{-R}\partial_R U}{2g^2} \right]^+, \label{pow}\\
		   \partial_t m(t,R) =&  \partial_R (pe^{-R}m) +\frac{\sigma^2}{2} \partial^2_{RR}(e^{-2R}m),\\
		   \bar{p}(t) =& \int\limits_{-\infty}^{\infty} e^R p(t,R)m(t,R)dR, \label{discreteImfg}	\\
		   \int_{-\infty}^{\infty} m(t,R)dR =&1, ~~~\text{where}~~~ m(t,R) \geq 0. \label{mDist}		   	  
	  \end{align}
  \end{mypro}

\begin{mylem}\label{MFGPower}
   	The average transmit power $\bar{p}(t)$ of a generic SBS is a derivative  of the average energy available with respect to time and can be calculated as 
   	\begin{equation}
  	 	\bar{p}(t) = -\frac{d}{dt}\int_{-\infty}^{\infty} e^{2R}m(t,R)dR.
   	\end{equation}
  \end{mylem}
  \begin{proof}
 See \textbf{Appendix C}. 
 \end{proof}
Since $\bar{p}$ is always non-negative,  the average energy in an SBS's battery is a decreasing function of time. That is, the distribution $m$ should shift toward  left when $t$ increases. This is because, we use the Wiener process in (\ref{mfg:1}). Since $dW_t$ has a normal distribution with mean zero, the energy harvested will be equal to the  energy leakage. Therefore, for the entire system, the total energy reduces when time increases.

 \begin{mylem}
	 If $(U_1,m_1)$ and $(U_2,m_2)$ are two solutions of \textbf{Proposition \ref{PDEs}} and $m_1=m_	2$, then we have $U_1=U_2$.
 \end{mylem}
 \begin{proof}
 First, from (\ref{mfg:2}) we derive Fokker-Planck equation: 
 	 \begin{align*}
	 	 \partial_t m_1(t,R) &=  \partial_R (p_1 e^{-R}m_1) +\frac{\sigma^2}{2}\partial^2_{RR}(e^{-2R}m_1), \\
	 	 \partial_t m_2(t,R) &=  \partial_R (p_2 e^{-R}m_2) +\frac{\sigma^2}{2} \partial^2_{RR}(e^{-2R}m_2).
 	 \end{align*}
Since $m_1=m_2=m$,  we subtract the first equation from the second one to obtain
$
\partial_R ((p_1-p_2)e^{-R}m)  = 0.
$
This means $(p_1-p_2)e^{-R}m$ is a function of $t$. Let us denote $f(t) = (p_1-p_2)e^{-R}m$, then we have $(p_1-p_2)m = f(t)e^R$.
From \textbf{Lemma \ref{MFGPower}}, if $m_1=m_2$ then $\bar{p}_1(t)=\bar{p}_2(t)$. Since $\bar{p}(t) = \int e^R pm dR$, we have
\begin{eqnarray}
\int_{-\infty}^{\infty} e^R p_1 m dR & = & \int_{-\infty}^{\infty} e^R p_2 m dR \nonumber \\
~\Rightarrow~
\int_{-\infty}^{\infty} e^R(p_1-p_2)m dR & = & 0, ~~~ \forall t.
\end{eqnarray}
Now we substitute
$(p_1-p_2)m = f(t)e^R$ that results in
$
		\int_{-\infty}^{\infty} f(t) dR = 0 ~~ \forall t.
$
This means $f(t)=0$ or $p_1=p_2$. Note that $U$ is a function of $\bar{p}$ and $p$. Since $\bar{p}_1=\bar{p}_2$ and $p_1=p_2$, it follows that $U_1 = U_2$.	This lemma confirms  that an SBS will act only against the mean field $m$. Thus only $m$ determines the evolution of the system. Two systems with the same mean field will behave similarly.
 \end{proof}

\subsection{Solving MFG Using Finite Difference Method (FDM)}

To obtain $U$ and $m$, we use the finite difference method (FDM) as in \cite{MFGlab} and \cite{bauso}. We discretize time and energy coefficient $R$ into large intervals as $[0,...,T_{max}\Delta t]$ and $[-R_{max}\Delta R,...,R_{max}\Delta R]$ with $\Delta t$ and $\Delta R$ as the step sizes, respectively. Then $U,m,p$ become matrices with size $T_{max} \times (2R_{max}+1)$. To keep the notations simple, we use $t$ and $R$ as the index for time and energy coefficient in these matrices with $t\in \{0,...,T_{max}\}$ and $R \in \{-R_{max},...,R_{max}\}$. For example, $m(t,R)$ is the probability distribution of energy $e^{R\Delta R}$ at time $t\Delta t$. Using the FDM, we replace $\partial_R U$, $\partial_t U$, and $\partial_{RR}^2 U$ with the corresponding discrete formulas as follows \cite{jichun}:
\begin{align}
\partial_t U(t,R) = &\frac{U(t+1,R)-U(t,R)}{\Delta t}, \\
\partial_R U(t,R) = &\frac{U(t,R+1)-U(t,R-1)}{2\Delta R}, \\
\partial_{RR}^2 U(t,R) = &\frac{U(t,R+1)-2U(t,R)+U(t,R-1)}{(\Delta R)^2}. 
\end{align}
By using  them in (\ref{valueF}) and after some simple algebraic steps, we have
\begin{equation}\label{discreteU}
\begin{split}
	U(t-1,R) =   U(t,R) + e^{-2R}\frac{\sigma^2{\Delta t}}{2(\Delta R)^2}A_1 
	- \Delta tB_1,
\end{split}
\end{equation}
where
\begin{equation*}
\begin{split}
	A_1 &= U(t,R+1)-2U(t,R)+U(t,R-1), \\
	B_1 &= \left(p(t,R)g\right)^2 
		 -\left(\bar{\lambda}\bar{p}(t)+\lambda N\right)^2.
\end{split}
\end{equation*}
Similarly, discretizing (\ref{discreteImfg}), we have
\begin{equation}\label{discreteM}
	\begin{split}
		m(t,R) =\frac{\Delta t}{2\Delta R}A_2 +\frac{\sigma^2\Delta t}{2(\Delta R)^2}B_2+  m(t-1,R),
	\end{split}
\end{equation}
where 
\begin{eqnarray}
	A_2 & =& e^{-(R+1)\Delta R}p(t-1,R+1)m(t-1,R+1)- \nonumber \\
			&  &e^{-(R-1)\Delta R}p(t-1,R-1)m(t-1,R-1), \nonumber\\
	B_2 &=& e^{-2(R+1)\Delta R}m(t-1,R+1)- \nonumber \\
      &  &2e^{-2R\Delta R}m(t-1,R)+e^{-2(R-1)\Delta R}m(t-1,R-1). \nonumber\\
\end{eqnarray}

To obtain $U,m,p$, and $\bar{p}$ using \textbf{Proposition \ref{PDEs}}, we need to have some boundary conditions. First, to find $m$, we assume that there is no SBS that has the battery level equal to or larger than $e^{R_{max} \Delta R}$ so that $m(t,R_{max})=0,~~\forall t$. This is true if we assume that $e^{(R_{max}-1) \Delta R}$ is the largest battery size of an SBS. Also, when $R=-R_{max}$, from the basic property of probability distribution 
\begin{equation}
	\begin{split}
	\sum_{R=-R_{max}}^{R_{max}} m(t,R) \Delta R  & = 1\\
	\Rightarrow m(t,-R_{max}) &= \frac{1}{\Delta R} - \sum\limits_{R=-R_{max}+1}^{R_{max}} m(t,R). \label{mBound}
	\end{split}
\end{equation}

Next, to find $U$, again we need to set some boundary conditions. Notice that $U(T_{max},R)=0$ for all $R$. We further assume the following:
\begin{itemize}
\item  Intutitively, if the battery level of a SBS is full, i.e., when $R=R_{max}$, this SBS should transmit something, or equivalently, $p(t,R_{max}) > 0$. That means
\begin{equation}\label{URmax}
	p(t,R_{max}) = \frac{\bar{\lambda}\bar{p}(t)+\lambda N}{g} + \frac{e^{-R_{max}\Delta R}\partial_R U(t,R_{max})}{2g^2}. 
\end{equation}
Therefore, if we know $U(t,R_{max}-1)$ and $p$, we can calculate $U(t,R_{max})$.

\item Similarly, it must be true that when the available energy is $0$, i.e., $R=-R_{max}$, an SBS will stop transmission. Therefore, we can assume
\begin{equation}\label{-URmax}
\begin{split}
	 \frac{\bar{\lambda}\bar{p}(t)+\lambda N}{g} + \frac{e^{R_{max}\Delta R}\partial_R U(t,-R_{max})}{2g^2} & = 0. 	
\end{split}
\end{equation}
Again, if we know $U(t,-R_{max}+1)$, we can calculate  $U(t,-R_{max})$. 

\item During simulations, in some cases when the density is very high, we  obtain very large (unrealistic) values of transmit power. Therefore, we must put an extra constraint for the upper limit. In this paper, we use $E(t) > p(t,E(t))\Delta T $, or $e^{R(t)\Delta R} \geq  p(t,R(t))\Delta T$, where $\Delta T$ is the duration of one time slot. This means, we have to limit the transmit power during one time step $\Delta t$ to be smaller than the maximum power that can be transmitted during one time interval $\Delta T$.
\end{itemize}
Based on the  above considerations, we develop an iterative algorithm (\textbf{Algorithm \ref{fdm}}) detailed as follows. 

 	\begin{algorithm}
\centering
  	\caption{Iterative algorithm for FDM}\label{fdm}
  	\begin{algorithmic}
  	
     \BState 	{\textbf{Initialize input}}

      \State\hspace{\algorithmicindent} Set up  $T_{max}\times(2R_{max}+1)$ matrices ${U}$, ${m}$, ${p}$,       
        and $T\times 1$ vector $\bar{p}$.
      \State\hspace{\algorithmicindent} Guess arbitrarily  initial values for power ${p}$, i.e., $p(t,R)= e^{R\Delta R}$.
      \State\hspace{\algorithmicindent} Initialize $i = 1$, $U(T_{max},.)=0 $, $m(0,.)=m_0(.)$, $m(t,R_{max})=0$, and  $p(t,0)=0$. 
      \State\hspace{\algorithmicindent} Initialize $\Delta R$ and $\Delta t$ as the step size of energy and time  
      with $(\Delta R)^2 > \Delta t$. 
      \State\hspace{\algorithmicindent} Set $\text{MAX}$ as the number of iteration.
      
  	\BState {\textbf{Solve PDEs with FDM}} 
  	\While{$i < \text{MAX}$}:
  	  
  \State\hspace{\algorithmicindent} Solve the Fokker-Planck equation  to obtain $m$ using (\ref{discreteM}) and (\ref{mBound}) with given $p$, $m_0$.
   	 \State\hspace{\algorithmicindent} Update $\bar{p}(t)$ for $T_{max} \geq t \geq 0$ using discrete form of equation (\ref{discreteImfg}).
  	
  	 \State\hspace{\algorithmicindent} Calculate $U$ for all $t<T_{max}$ by using (\ref{discreteU}), (\ref{URmax})  and (\ref{-URmax}) with $p$, $\bar{p}$.
  	 \State\hspace{\algorithmicindent} Calculate new transmission power $p_{new}$ using (\ref{pow}).  
  	 \State\hspace{\algorithmicindent}  Regressively update $p = ap + bp_{new}$ with $a+b=1$.
      \State\hspace{\algorithmicindent} \textbf{for} $R \in \{-R_{max},...,R_{max} \} $
	  \State\hspace{\algorithmicindent}\hspace*{0.5cm}    \textbf{if}  { $p(t,R) > \frac{e^{R\Delta R}}{\Delta T}$} \textbf{then}
   	  		 $p(t,R)=\frac{e^{R\Delta R}}{\Delta T}$. 	
   	   \State\hspace{\algorithmicindent} \textbf{end}
  	 \State\hspace{\algorithmicindent} $i \gets i+1$. 	
  	 	
  	\EndWhile
  	\State \textbf{end}
  	\BState {\textbf{Loop in $T_{max}$ time slots}}

  	  \State\hspace{\algorithmicindent}  At time slot $t$, SBS with energy battery $e^{R\Delta R}$ transmits with power $p(t,R)$. 
\end{algorithmic}
  	\end{algorithm}
  	
\subsection{Implementation of MFG}

For the MFG, we do not need the location information for each SBS. However, we need information about the average channel gain $\bar{g}$, the number of SBSs $M$ in one macrocell, and the initial distribution $m_0$ of the energy of the SBSs. Therefore, some central system should measure these information, solve the differential equations, and then broadcast the power policy $p$ to all the SBSs. It is more efficient than broadcasting all the information to all SBSs and let them solve the differential equations by themselves. Again, the central system only needs to re-calculate and broadcast to all SBSs a new power policy if there are changes in $\bar{g}$ or $M$. 

\section{Simulation Results and Discussions}

\subsection{Single-Controller Stochastic Game}

In this section, we quantify the efficacy of the developed stochastic policy in comparison to the simple Stackelberg game-based power control policy. The stochastic policy is obtained from the QCQP problem. On the other hand, for the Stackelberg policy, we follow a hierarchical method. 
Let us assume that at time $t$ SBS $i$ has $E_i$ Joules of energy in its battery. If the MBS transmits with power  $p_0$, then each SBS tries to transmit with a power such that 
\begin{equation*}
\begin{split}
\min_{p_1,p_2,...,p_M} \sum_{i=1}^{M} (p_ig_i - \lambda I_i)^2\\
\text{s.t.}~~~ p_i \Delta T \leq E_i 
 \end{split}
 \end{equation*} 
where $I_i$ is the interference from both the SBSs and the MBS. Since this is a convex problem, although it is solved independently by each SBS, the results are the same.  Notice that the objective function for each SBS is similar to the utility function in (\ref{u1}); therefore, it can provide a fair comparison against our method. The constraint means that each SBS cannot transmit more than the energy it has in its battery. Next, knowing that SBS $i$ will solve the above optimization problem to find its $p_i$,  the MBS calculates its ${\rm SINR}$ for different values of its transmit power and picks the optimal $p_0$. Knowing $p_0$, each SBS solves the convex optimization above to obtain its transmit power $p_1,...,p_M$. The tuple $(p_0,p_1,...,p_M)$ will be a Nash equilibrium because no one can choose a better option given the others' actions.

To solve the QCQP in (\ref{QCQP}), we use the \textit{fmincon} function from Matlab. 
Note that the \textit{fmincon} function may return a local optimal instead. Therefore, in order to obtain a good approximation for the optimal solution, in our simulations, we use the incremental method as described below.

First, our QCQP problem is stated as follows:
\begin{equation}
			\begin{split}
				&\max\limits_{\mathbf{m},\mathbf{x},\mathbf{\phi},\mathbf{\xi}} \mathbf{\pi}^{\text{T}} \mathbf{\phi}, \\
				\text{s.t.}~~~ & \mathbf{m}(R_0+R_1)\mathbf{x}-\mathbf{\pi}^{\text{T}} \mathbf{\phi} - \mathbf{1}^{\text{T}}\mathbf{\xi}=0, \\
				&{f(\mathbf{m},\mathbf{x},\mathbf{\phi},\mathbf{\xi}) \leq 0},
			\end{split}
\end{equation}
where $f()$ is a group of linear functions of $(\mathbf{m},\mathbf{x},\mathbf{\phi},\mathbf{\xi})$.

By solving this using \textit{fmincon}, we obtain a local optimal result $(\mathbf{m}_0,\mathbf{x}_0,\mathbf{\phi}_0,\mathbf{\xi}_0)$. Then we solve the updated QCQP problem as follows:
  \begin{equation}
  			\begin{split}
  				&\max\limits_{\mathbf{m},\mathbf{x},\mathbf{\phi},\mathbf{\xi}} \mathbf{\pi}^{\text{T}} \mathbf{\phi}, \\
  				\text{s.t.}~~~ & \mathbf{m}(R_0+R_1)\mathbf{x}-\mathbf{\pi}^{\text{T}} \mathbf{\phi} - \mathbf{1}^{\text{T}}\mathbf{\xi}=0, \\
  				&{f(\mathbf{m},\mathbf{x},\mathbf{\phi},\mathbf{\xi}) \leq 0}, \\
  				& \mathbf{\pi}^{\text{T}} \mathbf{\phi}_0 \leq \mathbf{\pi}^{\text{T}} \mathbf{\phi} + \epsilon,
  			\end{split}
\end{equation}
where $\epsilon$ is a small positive constant. By solving this new QCQP, we obtain a local optimal solution  that satisfies the constraints of the original QCQP and returns a better result. We keep repeating this step as long as \textit{fmincon} is able to return a solution $(\mathbf{m}_N,\mathbf{x}_N,\mathbf{\phi}_N,\mathbf{\xi}_N)$ with $N>0$ that satisfies the constraints. Then we say that $(\mathbf{m}_{N-1},\mathbf{x}_{N-1},\mathbf{\phi}_{N-1},\mathbf{\xi}_{N-1})$ is the optimal solution (with the error of $\epsilon$). Note that the optimal solution for this QCQP can be found by using brute force search through all possible integer values of $\mathbf{m}$.  Therefore, the result can be double checked when the number of SBSs is small.

In the simulations, the CES has a maximum battery size of $S=25$ and the volume of one energy packet is $K=25 \times 10^{-4}$ J. The duration of one time interval is $\Delta T = 5$~ms and the thermal noise is $N_0=10^{-8}$ W. The small cell users are considered to be in outage  if the ${\rm SINR}$ falls below 0.02. The MBS has two levels of transmission power [10; 20] watts and the SINR outage threshold is set to 5. The energy arrival at each SBS follows a Poisson distribution with unit rate.  The volume of each energy packet arriving at the CES is $C$ times larger than the energy packet collected by each SBS. Thus, the amount of energy in each packet at the CES will be $CK$ $\mu$J. Therefore, the CES should have a more efficient method to harvest energy than each SBS (in the case of Stackelberg method). However, as the maximum battery size of the CES is limited by $S$, its total available energy is always limited by the product $SCK$ $\mu$J regardless of $M$. For both the cases, each SBS can receive up to 1.5 mJ of energy from either the CES or the environment. At the beginning, the CES is assumed to have its battery full.

\begin{figure}[h]
\minipage{0.45\textwidth}
\centering
  \includegraphics[width=3in]{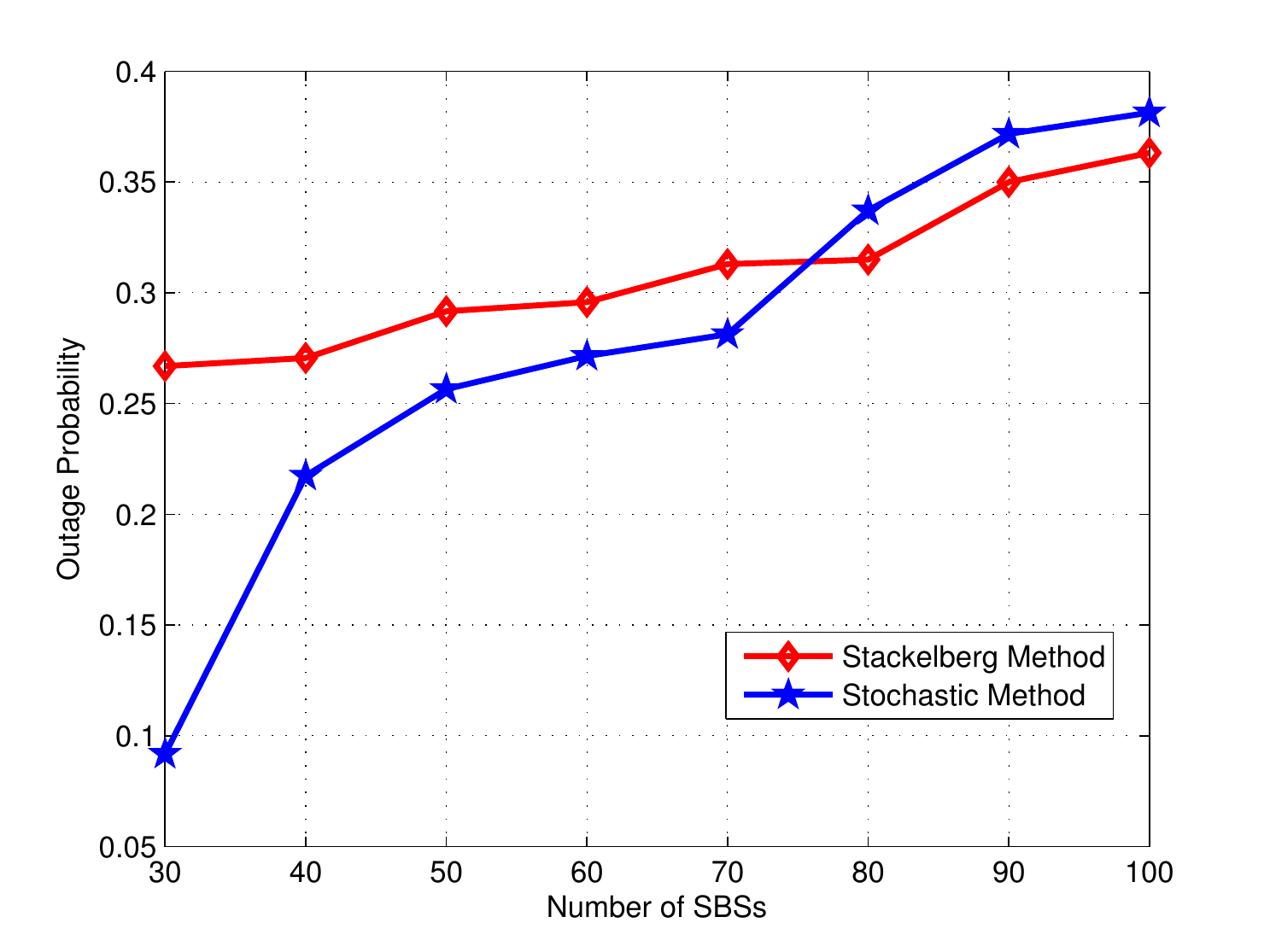}
  \caption{Outage probability of a small cell user with different number of SBSs when $S=25$ states, $C=60$, $\lambda_1=0.1$, $\lambda_0=10$.}
\label{density}
\endminipage\hfill
\minipage{0.45\textwidth}
\centering
  \includegraphics[width=3in]{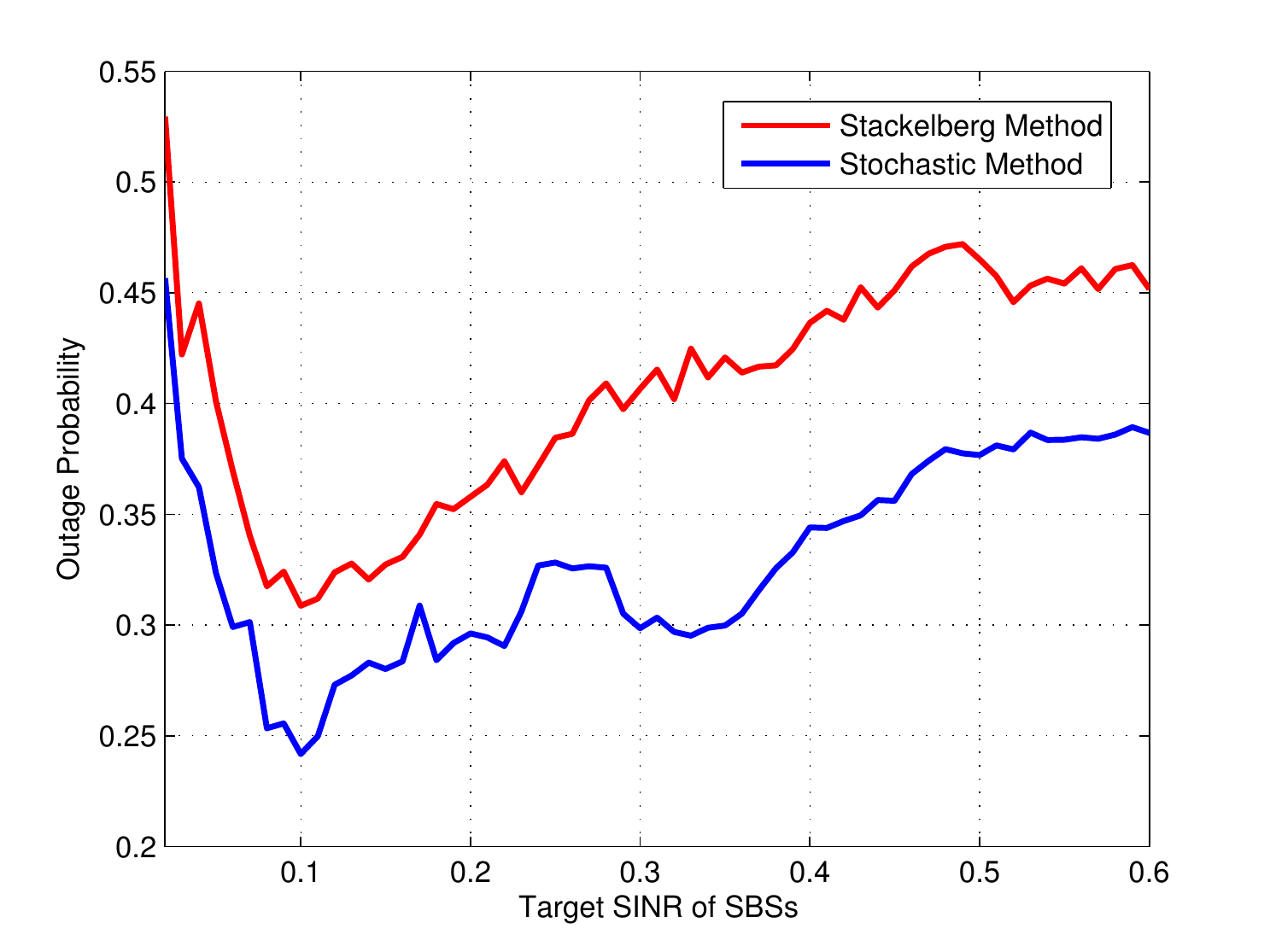}
\caption{Outage probability of a small cell user with different target ${\rm SINR}$ when $S=25$ states, $C=60$, $M=60$ SBSs, $\lambda_0=10$.}
\label{threshold}
\endminipage
\end{figure}

From Fig. \ref{density}, it can be seen that when the number of SBSs is smaller than some value, the stochastic method gives better results. This is because, the CES can redistribute the harvested energy among the SBSs based on their average channel gains, and also the QCQP in (\ref{QCQP}) gives a Nash equilibrium that favors the CES. However, at some point, its outage probability will be higher than that for the Stackelberg approach. This is not surprising since the CES can only store at most $S\times K \times C$ $\mu$J of energy. Therefore, when the number of SBSs increases, the average allocated power per SBS by the CES reduces  while the Stackelberg method  allows each SBS to harvest up to 1.5mJ no matter how large $M$ is. This means that the Stackelberg method can provide a better performance compared to using the CES when $M$ is large.

Following Fig. \ref{threshold}, by increasing the threshold SINR target $\lambda_1$ while keeping the number of SBSs fixed, we can reduce the outage probability of a user served by an SBS. This is understandable since the average ${\rm SINR}$ will approach the higher target and thus reduce the outage probability. However, for the stochastic 
methods, the outage will start to increase when the target ${\rm SINR}$ is larger than some value. To increase the average ${\rm SINR}$, the SBSs need to transmit with higher power to at least mitigate the cross-tier interference. However, a higher transmit power means a higher consumption of harvested energy, which can create shortages later. Also, higher transmit powers from the SBSs wil make the MBS to increase its own transmit power and thus create high cross-interference.  When the target ${\rm SINR}$ is higher than some value, the stochastic method will behave greedily by transmitting as much as possible and the outage will begin to increase. When the target ${\rm SINR}$ is large enough, the CES distributes all of the energy it currently has and thus the outage probability will become flat. A similar observation can be made for the Stackelberg approach. However, for this method, the energy cannot be redistributed to the SBSs with good channel gains to its users and the distribution of energy arrivals is ignored; therefore, the results are worse than those for our stochastic method. 

Fig. \ref{volume} shows the outage probability when increasing the quanta volume by choosing a higher multiplier $C$ for the CES. It is easy to see that, with a higher $C$, i.e., choosing a more effective method to harvest energy at the CES, we can achieve a better performance. The Stackelberg method does not use the CES, so the outage probability remains unchanged. Note that, since the battery size of each SBS is limited to 1.5 mJ, at some point, a higher $C$ does not improve the outage probability.

\begin{figure}[h]
\minipage{0.45\textwidth}
\centering
  \includegraphics[width=3in]{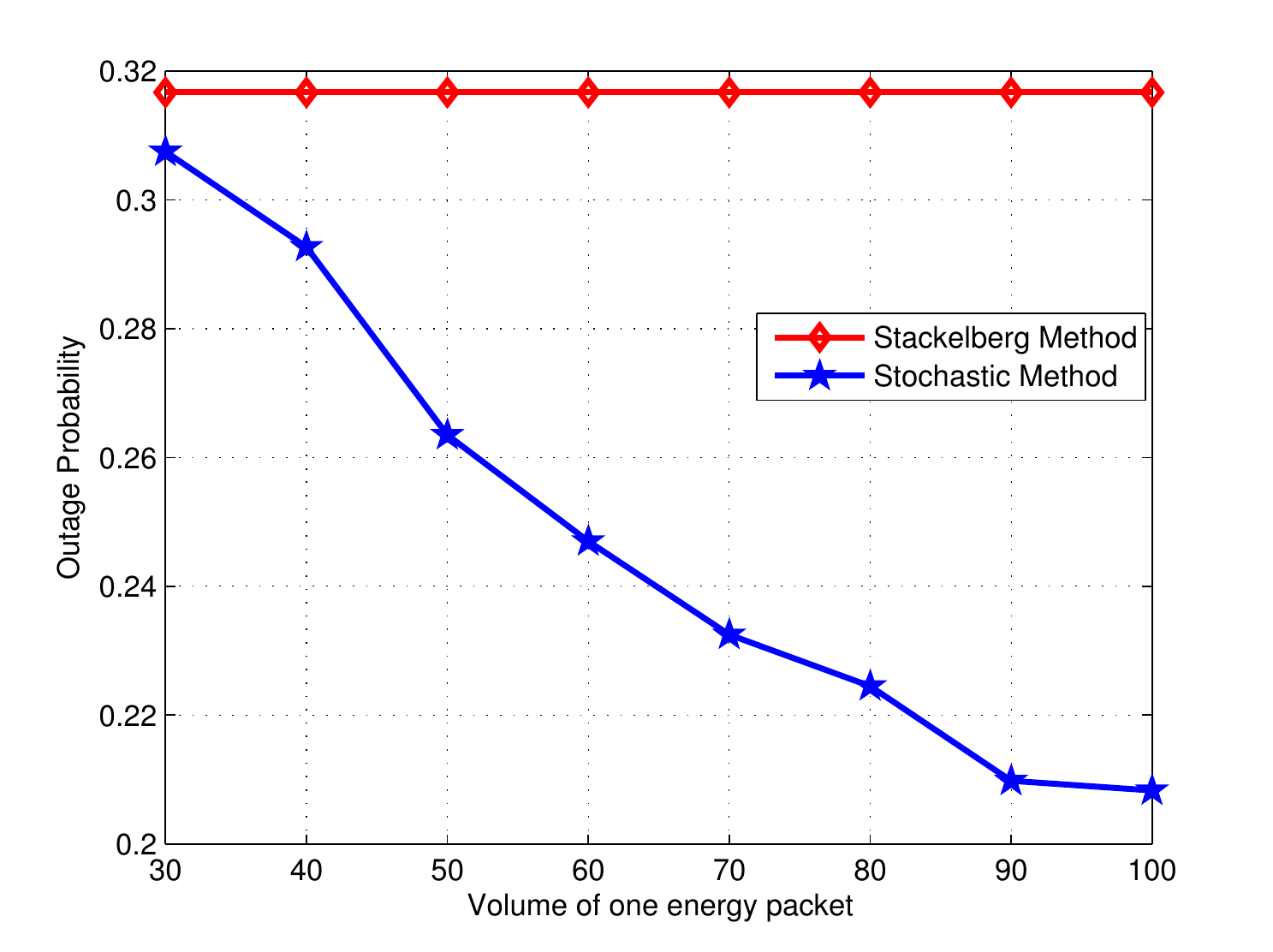}
  \caption{Outage probability of a small cell user  with different quanta volume when $S=25$ states, $M=60$ SBSs, $\lambda_1=0.1$, $\lambda_0=10$.}
	\label{volume}
\endminipage\hfill
\minipage{0.45\textwidth}
\centering
  \includegraphics[width=3in]{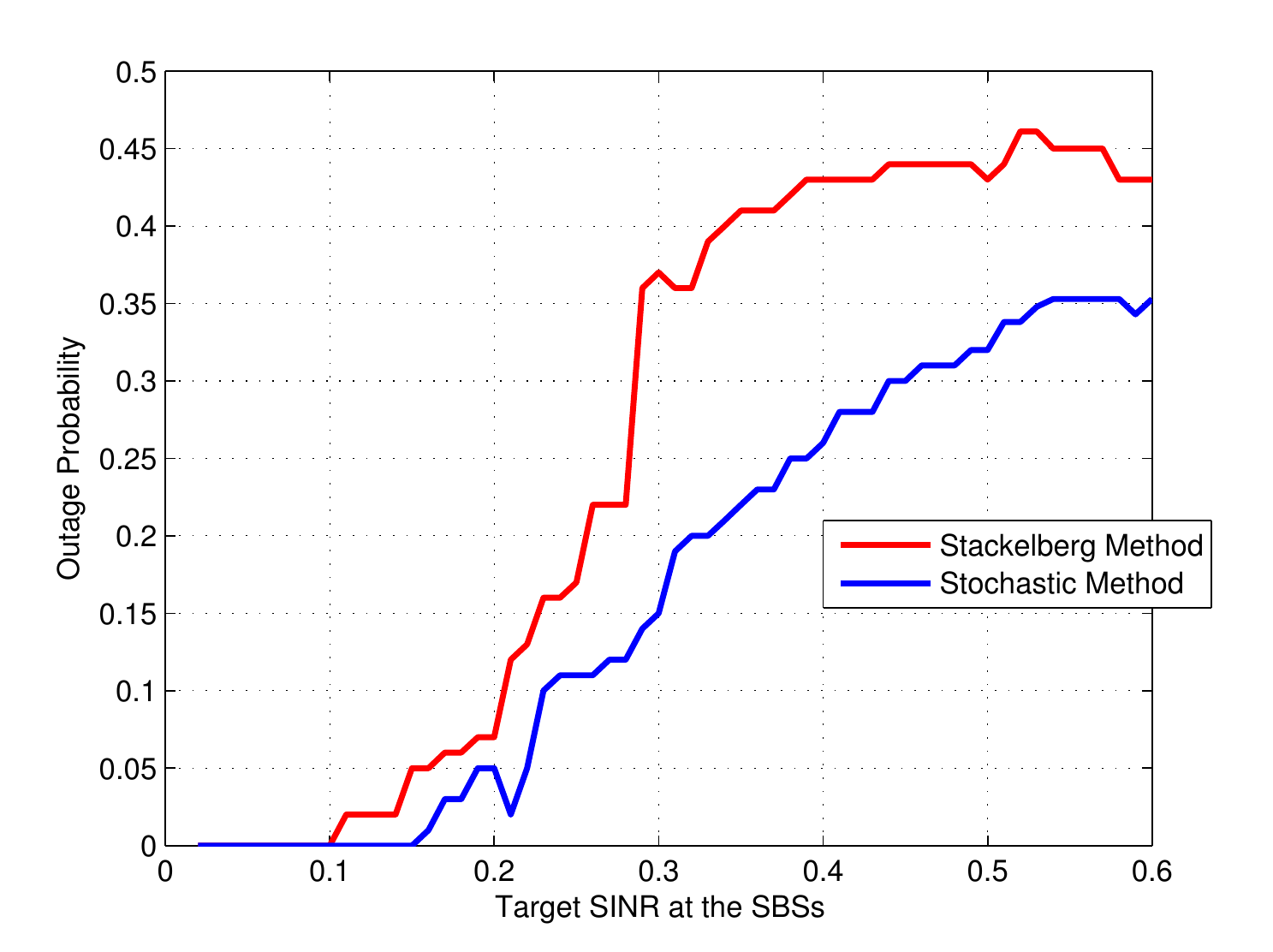}
\caption{Outage probability of a macrocell user when $S=25$ states, $M=60$ SBSs, $C=60$, $\lambda_0=10$.}
		\label{MBS_fbs}
\endminipage
\end{figure}

Fig. \ref{MBS_fbs} shows the outage probability of the macrocell user when $M=60$ and  $C=60$. The stochastic method gives better results in this case since the SBSs are more ``rational" in choosing their transmit powers in long term. Also, unlike the Stackelberg method, the CES has a fixed-energy battery, so when $M$ is large, the average amount of energy distributed to an SBS will be small, which in turn limits the cross-interference to the macrocell user. With the Stackelberg method, the SBSs only try to maximize their payoffs in the current time slot and ignore the distribution of energy; therefore, it uses a higher transmit power to compete against the MBS when the target ${\rm SINR}$ is increased; therefore, it creates a larger cross-interference and in turn increases the outage probability of the macrocell user.

In summary, we see that the centralized method using a CES can provide a better performance in terms of outage probability for both the MBS and SBSs. The advantages of using CES are two folds: First, it allows the harvested energy to be distributed to the SBSs which have good channel gains for the scheduled users, and second, it considers the probability distribution of energy arrivals when calculating the transmit power policy for both the MBS and SBSs. However, since the CES has a fixed battery size, this centralized model performs poorer when it needs to support a large number of SBSs. To improve this inflexibility, we can adjust other parameters as follows: change the target ${\rm SINR}$, increase the multiplier $C$, or increase the battery size of each SBS.

\subsection{Mean Field Game}


We assume that the transmit power at the MBS is fixed at 10W and it results in a constant noise at the user served by a generic SBS. The radius of the macrocell is $r=1000$ meter, so we have constant cross-interference $N_0=10^{-5}$ W. The target ${\rm SINR}$ is $\lambda = 0.002$ and assume that $g=\bar{g}=0.001$. We discretize the energy coefficient $R$ into 80 intervals, i.e., $R_{max}=40$ and $T_{max} = 1000$ intervals. Similar to the discrete stochastic case, each SBS can hold up to  150 $\mu$J in the battery, so the maximum transmit power is 30 mW. We impose the threshold such that an SBS will not transmit at $R=-R_{max}=-40$ or $E=0.6 ~\mu$J. The intensity of energy loss/energy harvesting, $\sigma$ is 1.

\begin{figure}[h]
\minipage{0.45\textwidth}
  \includegraphics[width=3in]{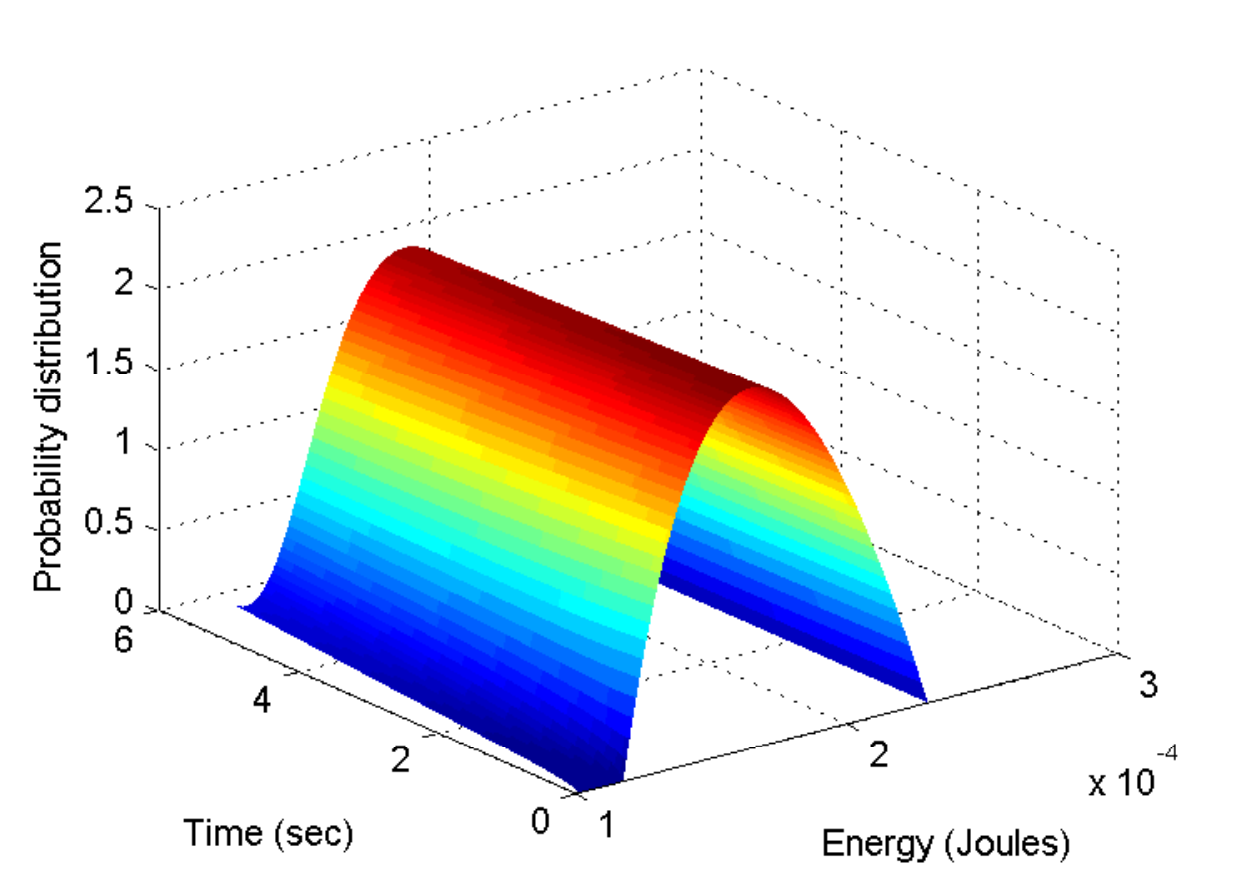}
  
  \caption{Energy distribution over time when $M=400$ SBSs/macrocell.}
\label{disMFG401}
\endminipage\hfill
\minipage{0.48\textwidth}
  \includegraphics[width=3in]{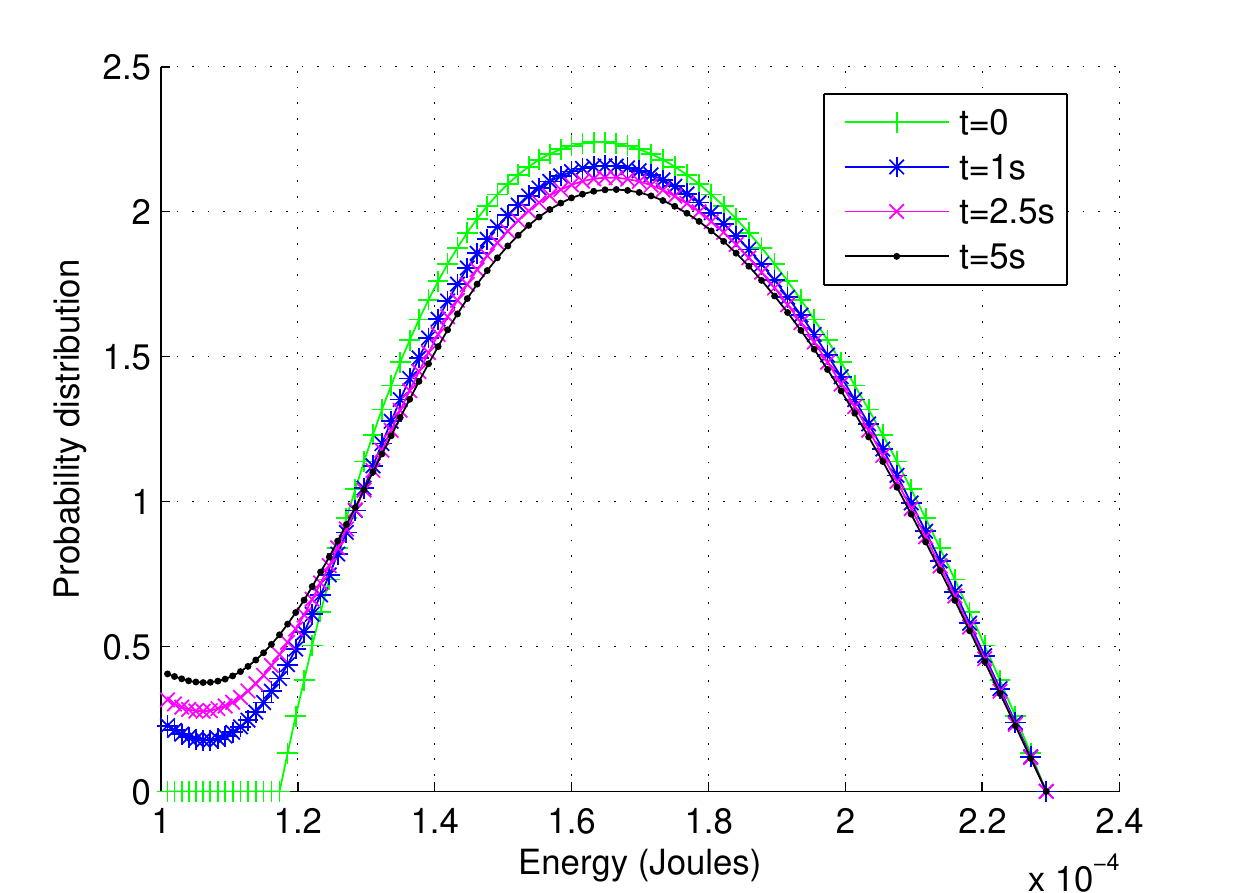}
  
\caption{Energy distribution over time when $M=400$  SBSs/macrocell.}
\label{disMFG402}
\endminipage\hfill
\end{figure}

\begin{figure}[h]
\minipage{0.48\textwidth}
  \includegraphics[width=3in]{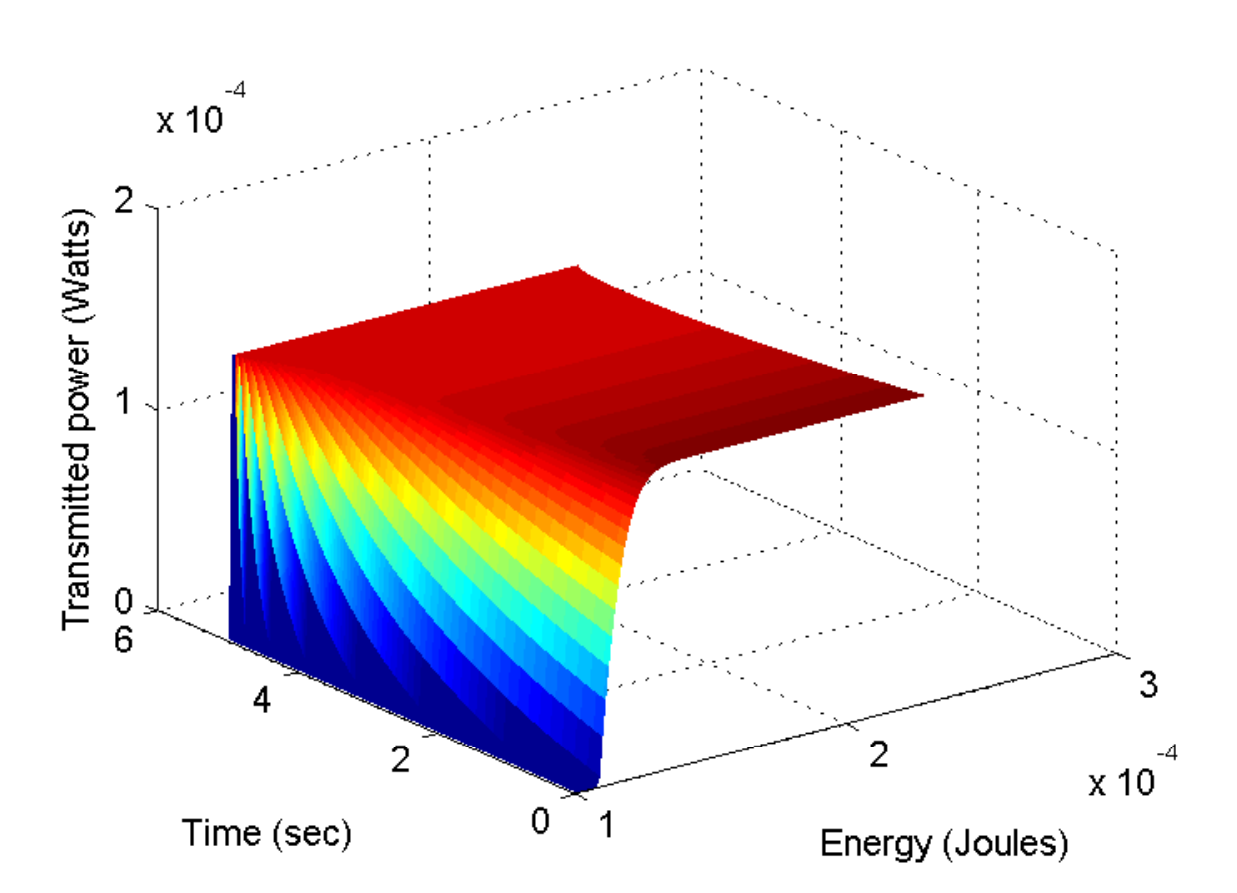}
  \caption{Transmit power to serve a generic user using MFG when $M=400$ SBSs/macrocell.}
\label{powMFG401}
\endminipage\hfill
\minipage{0.45\textwidth}
  \includegraphics[width=3in]{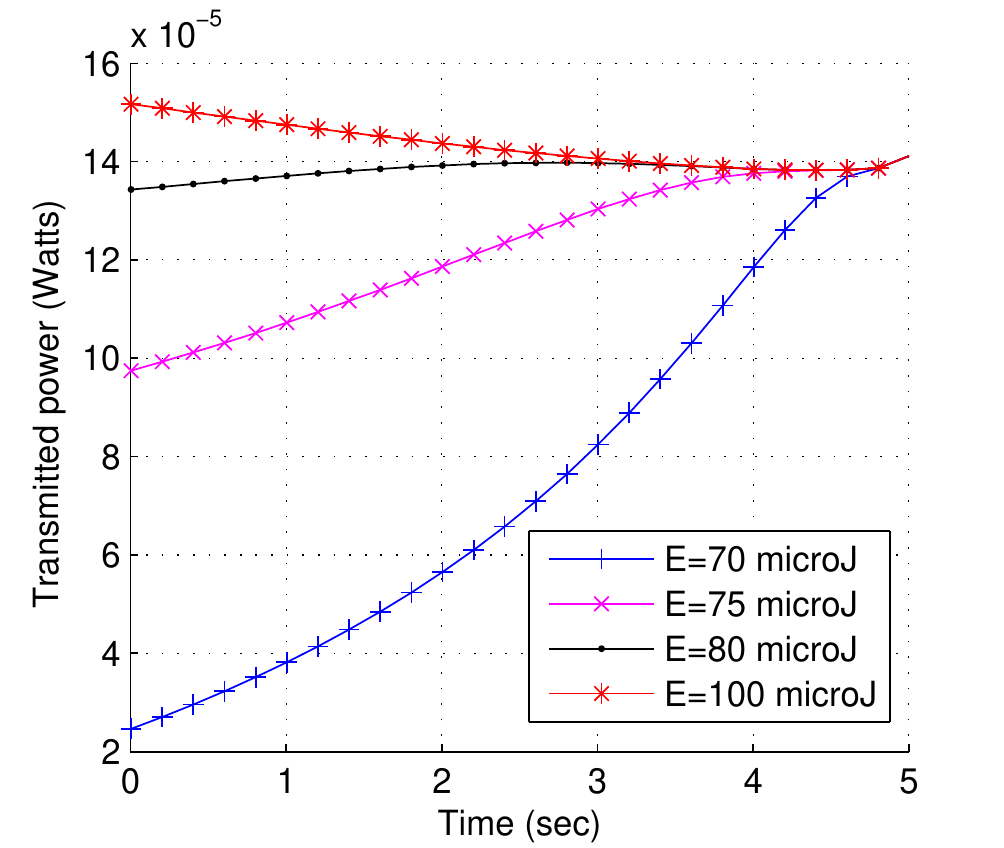}
\caption{Transmit power for different energy levels when $M=400$ SBSs/macrocell.}
\label{powMFG402}
\endminipage\hfill
\end{figure}

\begin{figure}[h]
\minipage{0.48\textwidth}
  \includegraphics[width=3in]{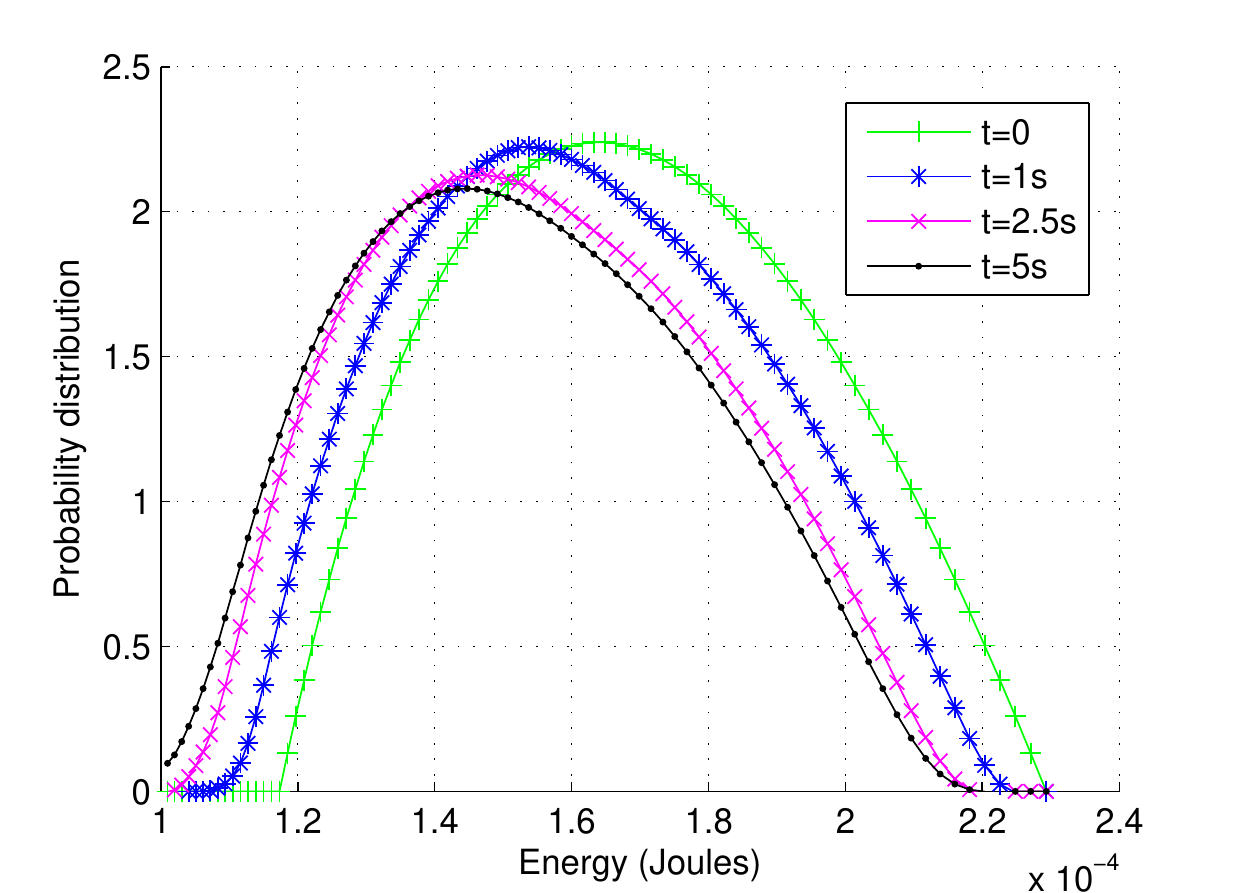}
\caption{Energy distribution over time when $M=500$ SBSs/macrocell.}
\label{disMFG502}
\endminipage\hfill
\minipage{0.48\textwidth}
  \includegraphics[width=3in]{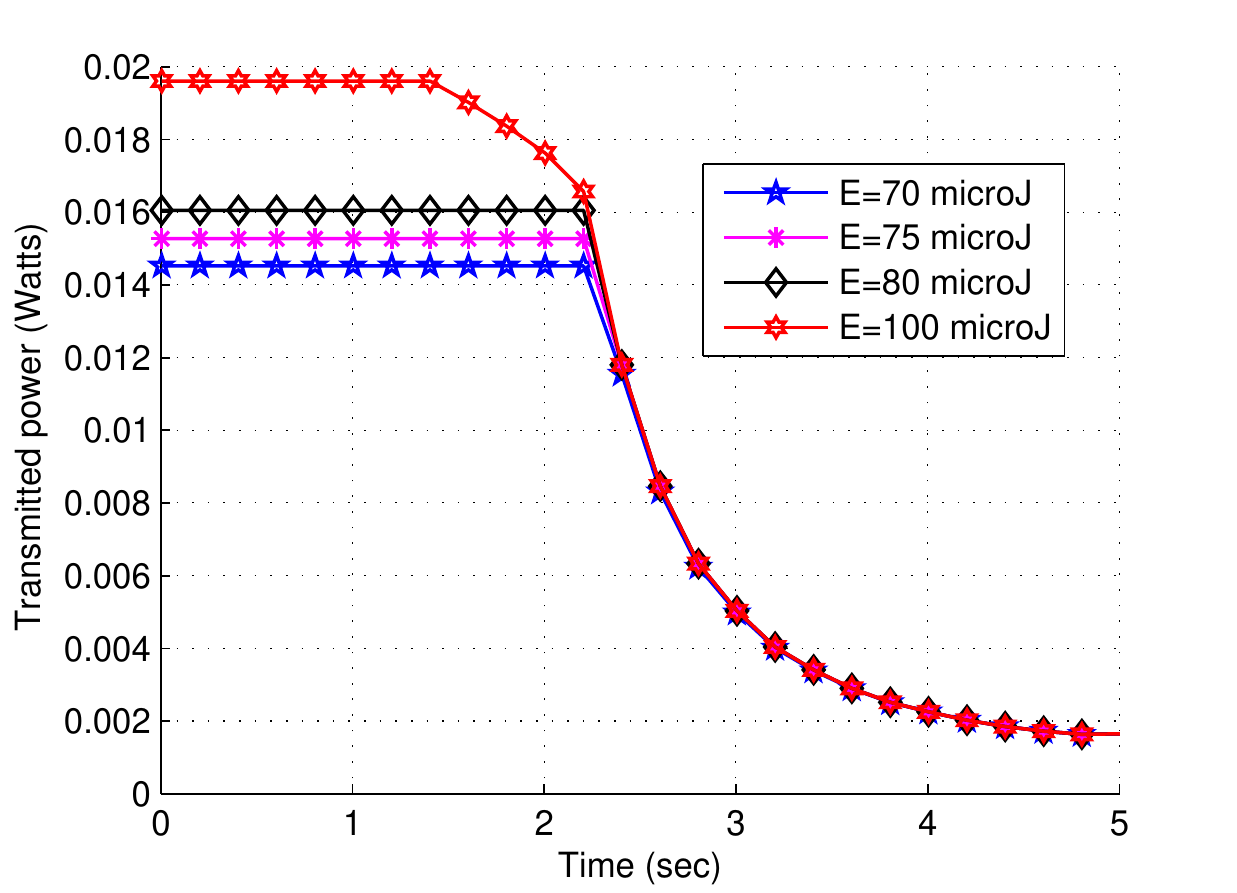}
\caption{Transmit power for different energy levels when $M=500$ SBSs/macrocell.}
\label{powMFG502}
\endminipage\hfill
\end{figure}

For $M=400$ SBSs/cell, we have $g = 0.001 > \bar{\lambda}= \lambda \bar{g} M = 0.0008$, so a generic SBS does not need to use a large amount of power in order to obtain the target ${\rm SINR}$. Notice that $\bar{p}$ is the average transmit power of a generic SBS. Therefore, if a generic SBS reduces $\bar{p}$, the cost term $\bar{\lambda}\bar{p}$ also reduces. Thus the difference between the cost and the received power $pg$ will be smaller, which is desirable. It makes sense that a generic SBS will try to reduce its power as much as possible in this case. The power cannot be zero though, because $N_0>0$. Moreover, from Fig. \ref{powMFG401} and Fig. \ref{powMFG402}, we see that, at the beginning, the SBS with higher energy (i.e., 100 $\mu$J) will transmit with a high power and will gradually reduce to some value. The SBSs with smaller battery will increase their transmit powers gradually. Since the transmit power is                                                                                                                                                                                       small, we see that in Fig. \ref{disMFG401} and Fig. \ref{disMFG402}, the energy distribution shifts to the left slowly.

On the other hand, when $M=500$ SBSs/macrocell, we have $g=\bar{{\lambda}}=0.001$. In this case, the effect is more complicated because reducing the transmit power may not reduce the gap between the received power $pg$ and the cost term $\bar{\lambda}\bar{p}+\lambda N$. Again, as can be seen from Fig. \ref{powMFG502}, the SBSs with larger available energy will transmit with large power first and after 
sometime when there is less energy available in the system, all of them start to use less power.  Therefore, as can be seen in Fig.~ \ref{disMFG502}, the energy distribution shifts toward the left with a faster speed than the previous case. 

For $M=600$ SBSs/macrocell, we have $g<\bar{{\lambda}}$. This means each SBS needs to transmit with a power larger than the average $\bar{p}$ to achieve the target ${\rm SINR}$. In Fig. \ref{powMFG602}, we see that the behavior of each SBS is the same as in the previous case. That is, the SBSs with higher energy transmit with larger power first and then reduce it, while the ``poorer" SBSs increase their transmit power over time. 

\begin{figure}[!htb]
\minipage{0.5\textwidth}
\includegraphics[width=3in]{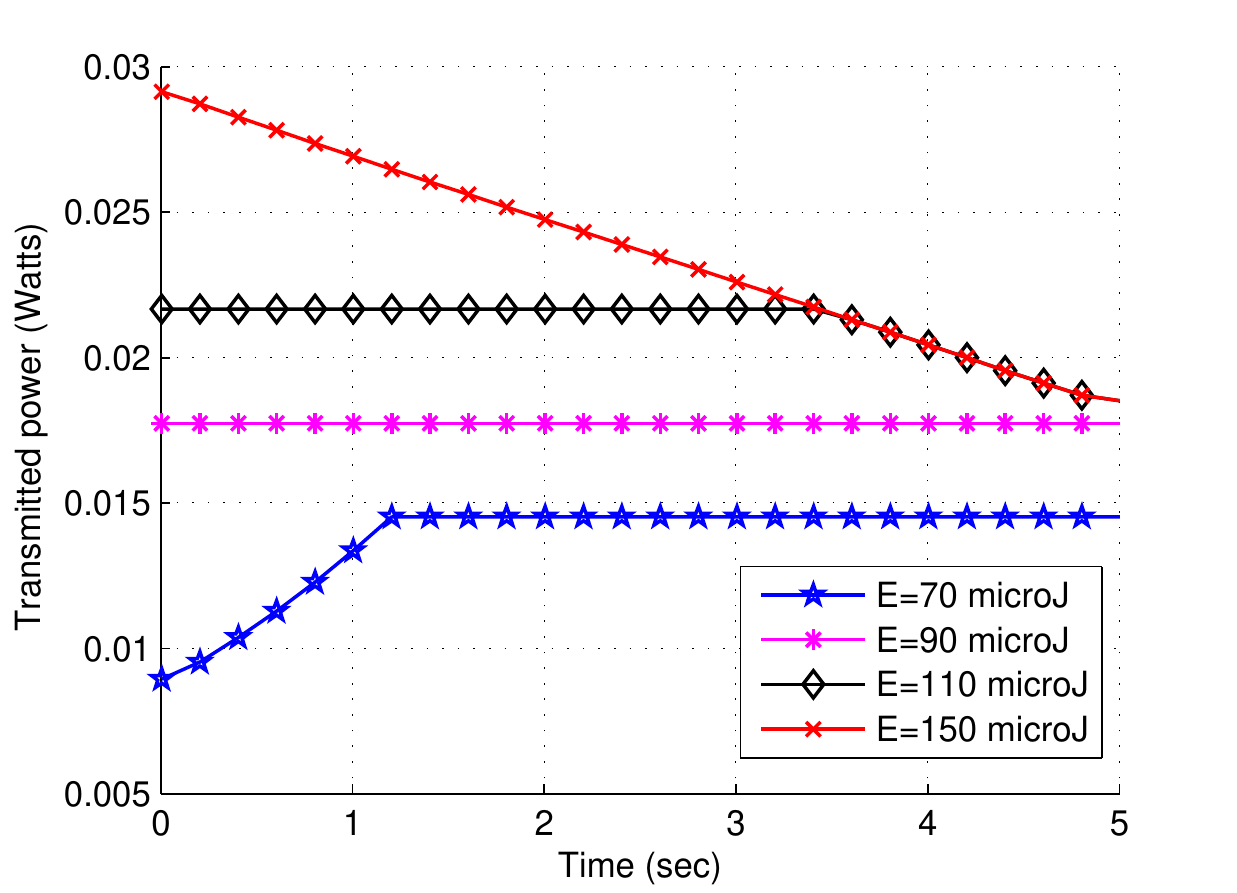}
\caption{Transmission power over time when $M=600$ SBSs/macrocell.}
\label{powMFG602}
\endminipage
\hfill
\minipage{0.5\textwidth}
\centering
\captionsetup{justification=centering}
\includegraphics[width=3in]{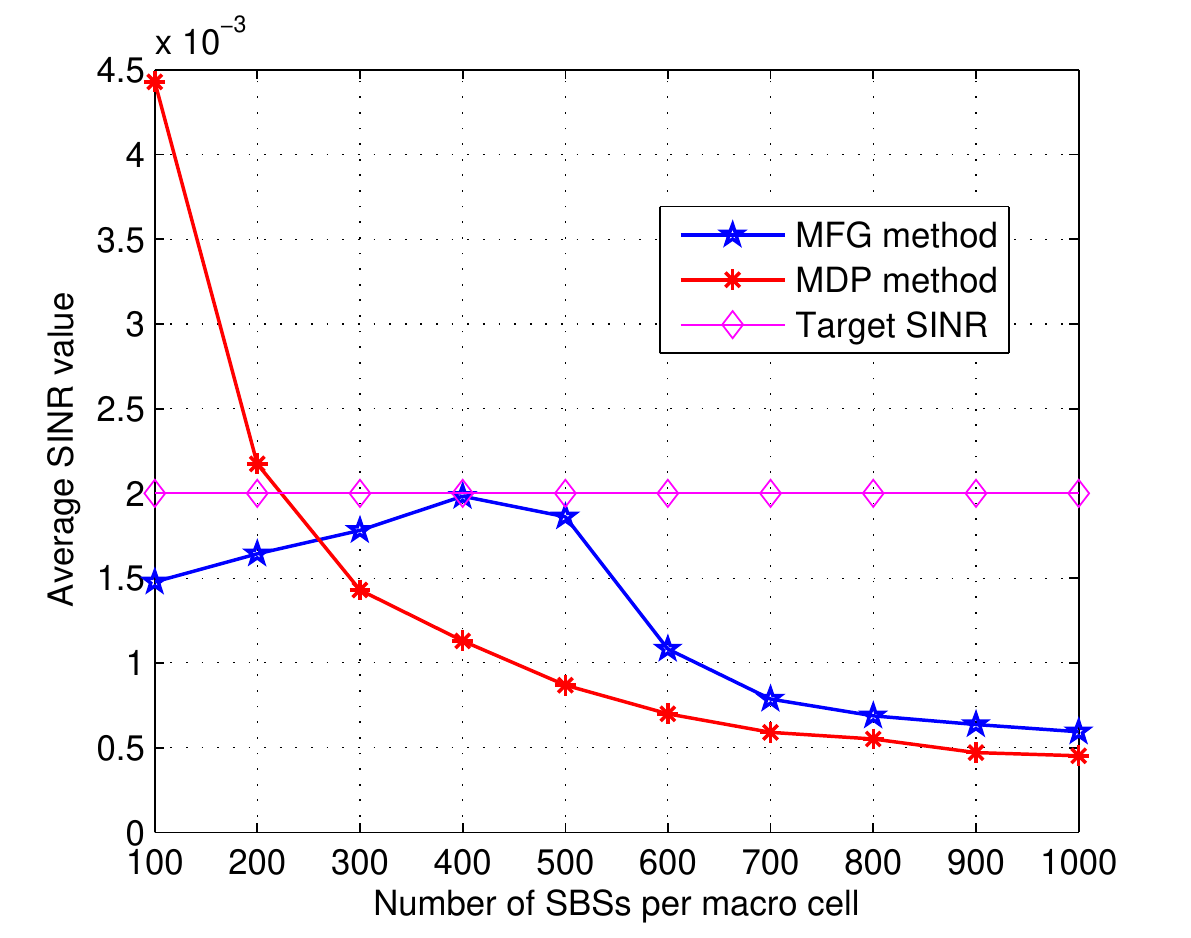}
\caption{Average SINR at a generic SBS.}
\label{MfgVsDiscrete}
\endminipage\hfill
\end{figure}

We compare the MFG model against the stochastic discrete model  for different values of  $M$. For simplicity, we assume that each SBS has the same link gain to its user as $g=0.001$. Also, we assume that the channel gain from each SBS to the user of another SBS is $\bar{g}=0.001$. Using \textbf{Remark~ \ref{rem2}}, it can easily be proven that in this case,  each SBS will transmit with the same power, i.e., if the CES sends $QCK$ $\mu$J of energy to $M$ SBSs, then each SBS receives $QCK/M$ $\mu$J. Then, the interference at each SBS will be calculated as $\frac{g}{(M-1)\bar{g}+N_0M\Delta T / (QCK) }$ with multiplier $C=20$ and the maximum battery size of the CES as $S=101$. Because the MBS is not a player of the game, the simulation step becomes simpler, and we only need to solve a linear program for the MDP problem instead of a QCQP. Therefore, we can call it as MDP method to accurately reflect the difference.

For the discrete stochastic case, we discretize the Gaussian distribution to model the  energy arrivals at the CES. 
The battery size of each SBS is still 150 $\mu J$. The average SINR of a generic small cell user using both the MFG and MDP models with different density is plotted in Fig.~\ref{MfgVsDiscrete}. We see that using the MFG model, the average ${\rm SINR}$ increases at the beginning and then it starts falling at some point. This is because, when the density is low, the interference from the MBS is noticeable (i.e., $10^{-5}$ W in our simulation). From the previous figures, it can be seen that an SBS will increase its power when the density is higher. Therefore, after some point the co-tier interference becomes dominant and the average ${\rm SINR}$ will begin to drop. It means at some value of the density, e.g., $M=400$ SBSs/macrocell in Fig.~\ref{MfgVsDiscrete}, we obtain the optimal average ${\rm SINR}$. We notice that the MFG model performs better than the MDP model with the CES. This is due to the limited battery size of CES which makes it difficult to support a large number of SBSs.

In summary, we have two important remarks for the MFG model. First, if the density of small cells is high, the SBSs will transmit with higher power. Second, from  Fig.~\ref{MfgVsDiscrete}, we see that by choosing a suitable density of the SBSs, we can obtain the highest average ${\rm SINR}$. From (\ref{pow}), it can be easily proven that the average ${\rm SINR}$ at a user served by an SBS will always be smaller than the target ${\rm SINR}$ (because $\partial_R U < 0$). Therefore, the highest average ${\rm SINR}$ is also the closest to the target ${\rm SINR}$, which is our objective in the first place.

\section{Conclusion}

We have proposed a discrete single-controller discounted two-player stochastic game to address the problem of downlink power control in a two-tier macrocell-small cell network under co-channel deployment where the SBSs use stochastic renewable energy source. For the discrete case, the strategies for both the macrocell and SBSs have been derived by solving a quadratic optimization problem. The numerical results have shown that these strategies can perform well in terms of outage probability experienced by the users. We have also applied a mean field game model to obtain the optimal power for the case when the number of SBSs is very large. We have also discussed the implementation aspects of these models in a practical network. 

In this paper, we have not explicitly considered the correlation in the energy arrival process.  However, this correlation can be modeled by assuming that the energy arrival has Markovian property. In this case, to calculate the transition probability, we will need to extend the definition of the state to a two-element vector, one is the current energy in the storage and the second is the energy arrival at this time slot. Moreover, we have not considered the details of cost and latency analysis related to the information exchange to and from the CES and also the charge and discharge loss of the battery storage. These issues can be addressed in future.


\appendices

 
\renewcommand{\theequation}{A-\arabic{equation}}  
\setcounter{equation}{0}  
\section{} 

Denote the distance between the BS $B$ to its user $D$ as $BD=a$. If $D$ is uniformly located inside the disk centred at $B$, the PDF of $BD$ is $f_D(BD=a) = \frac{2a}{r^2}$. Denote by $\theta$ the value of the angle $\angle{ABD}$, $\theta$ is uniformly distributed between $(0,2\pi)$. Using the cosine law
$d^2 = R^2 + a^2 - 2aR \cos(\theta)$,
	we obtain
	\begin{align}
	\mathbb{E}[d^{-4}]  = \int_0^{2\pi}\int_0^r (R^2 + a^2 - 2aR \cos\theta)^{-2} \frac{1}{2\pi} \frac{2a}{r^2} \, da \, d\theta. 
	\end{align}
	First, we solve the indefinite integral over $\theta$ as 
$
		\int (R^2 + a^2 - 2aR \cos\theta)^{-2}d\theta = f_1(a,\theta) + f_2(a,\theta) + L,
$
	where $L$ is a constant and
	\begin{align*}
		f_1(a,\theta) = &\frac{2(R^2+a^2)}{(R^2-a^2)^3}\arctan\frac{(R+a)\tan{\frac{\theta}{2}}}{R-a}, \quad \mbox{and} \\
		f_2(a,\theta) = 
		&\frac{2aR\sin{\theta}(R^2+a^2-2aR\cos\theta)}{(R^2-a^2)^2}.
	\end{align*}
Since $\sin 0 = \sin 2\pi = 0 $,
after integrating $f_2$ over $[0,2\pi]$, we can ignore it. Thus
	\begin{equation*}
		\int_{0}^{2\pi} (R^2 + a^2 - 2aR\cos\theta)^{-2}d\theta = {\pi}\frac{2(R^2+a^2)}{(R^2-a^2)^3}.
	\end{equation*} 
Next, we integrate the above result over  $a$ to obtain the indefinite integral as:
	\begin{equation}
		\frac{1}{r^2}\int a\frac{2(R^2+a^2)}{(R^2-a^2)^3}da	 = \frac{1}{r^2}\frac{a^2}{(R^2-a^2)^2} + L.
	\end{equation}		
Applying the upper and lower limits of $a$, we complete the proof.

\renewcommand{\theequation}{B-\arabic{equation}}    
  \setcounter{equation}{0}  
\section{Proof of Theorem 2}

First, we prove that given $\mathbf{n}$, there exists a  pure stationary strategy $\mathbf{m}$ which is the best response  of the MBS against $\mathbf{n}$. Since the action set of the MBS is fixed, at each state $s$, given strategy $\mathbf{n}(s)$ of the CES, the MBS just needs to choose a mixed stationary strategy $\mathbf{m}(s)$ such that its average payoff is maximized.
At state $s$, the average utility function of the MBS is 
\begin{equation}
	\mathbb{E}[U_1] = \sum\limits_{p_0^s \in \mathcal{P}}\sum\limits^{s}_{j=0} - (p_0^s\bar{g}_0-\lambda_0\bar{I}(p_0^s,j))^2 \mathbf{m}(s,p_0^s)\mathbf{n}(s,j),
\end{equation}
where $p_0^s$ and $j \in \{0,...,s\}$ are the transmit power of the MBS and the number of energy packets distributed at the CES at state $s$, respectively. $\bar{I}(p_0^s,j)$ is the average interference from other SBSs to the MBS if the CES distributes $j$ energy packets and the MBS transmits with power $p_0^s$. We have
$
		\bar{I}(p_0^s,j) = \sum\limits^{M}_{i=1}p_i\bar{g}_{0,i}+N,
$ where $(p_1,p_2,...,p_M)$ is the solution of  \textbf{Remark~ \ref{rem2}}, with $Q$ and $p_0$ replaced by $j$ and $p_0^s$, respectively.
 Since $\sum_{p_0^s \in \mathcal{P}}\mathbf{m}(s,p_0^s)=1$ and $\mathbf{m}_s$ is a non-negative vector, we have
\begin{equation}\label{maxMBS}
		\mathbb{E}[U_1] \leq \max\limits_{p_0^s \in \mathcal{P}} \left\{ -\sum\limits^{s}_{j=0} (p_0^s\bar{g}_0-\lambda_0\bar{I}(p_0^s,j))^2	\mathbf{n}(s,j) \right\}. 
\end{equation}

Since the set $\mathcal{P}$ is fixed and finite, there always exists at least one value of $p_0^s$ that achieves the maximum for the right hand side. That means when the game in state $s$, the MBS can choose this power level with probability of 1. However, obtaining a closed-form $p_0^s$ is difficult because first we need to find  $(p_1,...,p_M)$ in closed-form by solving $(\ref{dist})$.

Nevertheless, if the average channel gains from each SBS to the macrocell user (say $\bar{g}_{0,SBS}$ ) are same, we can obtain $p_0^s$ in closed form by defining
$
	\bar{I}(p_0^s,j) = \sum\limits_{i=1}^{M}p_i\bar{g}_{0,SBS}+N_0 = \bar{g}_{0,SBS}\sum_{i=1}^{M}p_i+N_0
	= \frac{K}{\Delta T}j\bar{g}_{0,SBS}+N_0.
$
The final equality is from (\ref{power}). Replacing this result back into (\ref{maxMBS}), we have 
\begin{equation}
\begin{split}
	\mathbb{E}[U_1] \leq \max\limits_{p_0^s \in \mathcal{P}}  -\sum\limits^{s}_{j=0} (p_0^s\bar{g}_0-\lambda_0(\frac{K}{\Delta T}j\bar{g}_{0,SBS}+N_0))^2
			\mathbf{n}(s,j). 	
			\end{split}
\end{equation}
The right hand side of this inequality is a strictly concave function (downward parabola) with respect to  $p_0^s$. Note that $\sum_{j=0}^{s} \mathbf{n}(s,j)=1$. The parabola will achieve the maximum value at its vertex given by
\begin{equation}
 p_0^{s*} = \frac{\lambda_0 \sum_{j=0}^{s}\left(\frac{K}{\Delta T}\bar{g}_{0,SBS}j + N_0\right)\mathbf{n}(s,j)}{\bar{g}_0}.
 \end{equation}
If $p_0^{s*}$ is not available in $\mathcal{P}$, since the right hand side of the inequality above is a parabola w.r.t. $p_0^{s}$, the best response $p_0^{s}$ to $\mathbf{n}(s)$ is the one nearest to the vertex $p_0^{s*}$.

On the other hand, given strategy $\mathbf{m}$ of the MBS,  the problem of finding the best response strategy $\mathbf{n}$ for the CES is simplified into a simple MDP in (\ref{P}). Then, there always exists a pure stationary strategy $\mathbf{n}$ \cite[Chapter~ 2]{filar3}. This completes the proof.
 
\renewcommand{\theequation}{C-\arabic{equation}}    
  \setcounter{equation}{0}  
\section{Proof of Lemma \ref{MFGPower}}
Using the stochastic differential equation in (\ref{mfg:1}) at time $t$,
$
  		dE(t) = -p(t,E(t))dt + \sigma dW_t,
$
  	we obtain the integral form as follows:
  	\begin{equation}\label{integralForm}
  	\begin{split}
 		E(t+t') -E(t) &= -\int_{t}^{t+t'}p(v,E(v))dv + \int_{t}^{t+t'}\sigma dW_v \\
  		&= -p(\bar{t},E(\bar{t}))t' + \sigma (W_{t+t'}-W_{t}),
  	\end{split}
  	\end{equation}
  	where $\bar{t} \in (t,t+t')$. We obtain the second equality using the mean value theorem for integrals: If $G(x)$ is a continuous function and $f(x)$ is integrable function that does not change sign on the interval $[a,b]$, then there exists $x \in [a,b]$ such that $
	  	\int_{a}^{b} G(t)f(t) dt = G(x)\int_{a}^{b}f(t)dt.
$
  	Since equation (\ref{integralForm}) is true for all SBSs, taking expectation of this equality above for all SBSs (or all possible values of $E$), we have
  	\begin{align}
  			\mathbb{E}[E(t+t')] -\mathbb{E}[E(t)] &=  -\mathbb{E}[p(\bar{t},E(\bar{t}))]t' + \sigma\mathbb{E}[W_{t+t'}-W_{t}], \\
  			\int\limits_{0}^{\infty}E m(t+t',E)& dE - \int\limits_{0}^{\infty}Em(t,E)dE   		
  			= -t'\mathbb{E}[p(\bar{t},E(\bar{t}))],
  	\end{align}
where $m(t,E)$ is the distribution of $E$ in the system at time instant $t$. Using the fact that $W$ is a Wiener process, $W_{t+t'}-W_{t}$ follows a normal distribution with mean zero (i.e., $\mathbb{E}[W_{t+t'}-W_{t}]=0$).

By dividing both sides by $t'$ and letting $t'$ to be very small (or $t' \rightarrow dt$), we have $\bar{t} \rightarrow t$ and 
  	\begin{equation}
  		\frac{d\int\limits_{0}^{\infty}Em(t,E)dE}{dt} = - \left(\int\limits_{0}^{\infty}p(t,E)m(t,E)dE\right) = -\bar{p}(t).
  	\end{equation}
Using  $m(t,R)=m(t,E)$, $dE = e^RdR$, and changing the variable $E$ to $R$, we complete the proof.


\end{document}